\documentclass[aps,prd,nofootinbib,english,10pt]{revtex4-2}
\usepackage[utf8]{inputenc}
\usepackage{graphicx,color} 
\usepackage{xcolor}
\usepackage{rotating}
\usepackage{amssymb}
\usepackage{amsfonts}
\usepackage{amsmath}
\usepackage{xspace}
\usepackage{mathtools} 
\usepackage{verbatim}
\usepackage{comment}
\usepackage{enumitem}
\usepackage{physics}
\usepackage{graphicx}
\usepackage[colorlinks]{hyperref}
\usepackage{MnSymbol}
\usepackage{graphicx}
\usepackage{accents}
\usepackage[ignoreall]{geometry}

\newtheorem{theorem}{Theorem}

\newtheorem{lemma}[theorem]{Lemma}

\newtheorem{proposition}[theorem]{Proposition}

\newenvironment{proof}[1][Proof]{\noindent\textbf{#1.} }{\ \rule{0.5em}{0.5em}}
\newcommand{\NV}[1]{{\textcolor{black}{ #1}}}
\newcommand{\CP}[1]{{\textcolor{black}{ #1}}}

\begin{document}
	
	\title{On spherically symmetric Berwald vacuum solutions in Finsler gravity}
	
	\author{N. Voicu}
	\email{nico.voicu@unitbv.ro}
	\affiliation{Department of Mathematics and Computer Science, Transilvania University of Brasov, Romania}
	
	\author{D.M. Birla}
	\email{diana.birla@unitbv.ro}
	\affiliation{Department of Mathematics and Computer Science, Transilvania University of Brasov, Romania}
	
	\author{C. Pfeifer}
	\email{christian.pfeifer@zarm.uni-bremen.de}
	\affiliation{ZARM, University of Bremen, 28359 Bremen, Germany}
	\affiliation{University of Rijeka, Faculty of Physics, 51000 Rijeka, Croatia.}
		


\begin{abstract}
Berwald-Finsler spacetimes are Finsler spacetimes that are closest to pseudo-Riemannian geometry, as their canonical nonlinear connection defines an affine connection on spacetime. In this paper, we find the first exact, non-Ricci flat, $SO(3)$-symmetric Berwald solutions to the Finsler gravity vacuum equation. To this aim, we first select from all spherically symmetric Berwald spacetimes, the only class that \textcolor{black}{admits flat Finsler spacetime structures, making it particularly suitable for future investigations of asymptotic flatness.} Then, for this class, we completely solve the Finsler gravity vacuum equation, and find three families of solutions. In particular, we show that in Finsler geometry there exist $SO(3)$-symmetric, vacuum solutions that are not Ricci-flat. These solutions are promising candidates to model the gravitational field around compact objects, beyond their Riemannian description.
\end{abstract}

\maketitle

\tableofcontents

\section{Introduction}
 \CP{Spatially spherically symmetric spacetimes} describe the gravitational field of non-rotating compact objects and serve as foundation to find the gravitational field of realistic, axially symmetric, rotating astrophysical systems such as black holes, ordinary and neutron stars, black hole-mergers, dust clouds and many more. As long as the the objects are slowly rotating and can be treated perturbatively, the spherically symmetric spacetime approximation is sufficient to predict observables from these systems with good accuracy. 

In general relativity, the Jebsen-Birkhoff Theorem \cite{VojeJohansen:2005nd, 1923Birk} ensures that the only spatially spherically symmetric vacuum solution of the Einstein equations must be the (asymptotically flat, static) Schwarzschild solution. As the Einstein vacuum equations are identical to the vanishing of the Ricci curvature tensor, one can rephrase the statement of the theorem as: the only  spatially spherically symmetric vacuum solutions are given by the family of Ricci-flat Schwarzschild spacetimes (parametrized by the Schwarzschild radius). Yet, in theories of gravity that go beyond general relativity, this statement is not necessarily true. 

In this article, we consider Finsler gravity, based on pseudo-Finsler spacetime geometry, see for example for extensive recent reviews \cite{Heefer:2024kfi,Sanchez:2025ifx}. In the Finslerian setting, the geometry of spacetime is derived from a (pseudo-)norm of tangent vectors -- which defines the length of curves in a largely similar way as the geometry of spacetime is derived, in general relativity, from the spacetime metric of pseudo-Riemannian geometry \cite{Bao, Miron, Finsler}. Recent results indicate that Finsler geometry is a viable candidate to improve the description of the gravitational field beyond general relativity and has drawn an increasing attention from both mathematics and theoretical physics communities \cite{Beem,Javaloyes:2018lex,Javaloyes:2022hph,Carvalho:2022sdz, Fuster:2015tua, Heefer_2023, Li:2015sja, Minas:2019urp, Papagiannopoulos:2017whb, Pfeifer:2019, Saridakis:2021vue, Stavrinos2014,Tavakol2009, Triantafyllopoulos:2018bli}.  Specific applications of Finsler geometry that are particularly to be highlighted are: 

\begin{itemize}
	\item the description of the gravitational field of kinetic gases through Finsler geometry \cite{Hohmann:2019sni,Hohmann:2020yia}, which has led to the finding of an accelerated expansion of the universe without dark energy~\cite{Pfeifer:2025tda};
	\item the interpretation, in metric-affine geometry, of autoparallels that arise from parametrization-invariant actions, as geodesics on a Finsler spacetime \cite{Fuster:2020upk,Csillag:2026kdc};
	\item  the geometric description of quantum gravity phenomenology effects including deformations of local Lorentz invariance \cite{Addazi:2021xuf, Gibbons:2007iu, Kostelecky:2011qz, Lobo:2020qoa, Mavromatos:2010nk, Zhu:2022blp,Zhu:2023kjx};
	\item propagation of waves in media and on curved spacetimes \cite{Werner:2012rc,Caponio:2026ida}.
\end{itemize}


There have been several proposals for suitable field equations that determine the Finslerian geometry of spacetime dynamically, of which the physically best motivated and most promising one is the action-based scalar field equation which naturally couples the kinetic gas to Finsler geometry proposed in \cite{Hohmann:2019sni},\cite{Hohmann:2020yia} -- or its Palatini generalization -- see also \cite{Hohmann:2021zbt,Javaloyes2022PalatiniFinsler,Sanchez:2025smd} for a detailed discussion. To describe compact objects, it is important to classify and understand spatially spherically symmetric vacuum solutions of the Finsler gravity equations.\vspace{12pt}

Among all Finsler spacetimes, the ones of Berwald type \cite{Berwald1926,Szilasi2011,Szabo,Crampin,Fuster:2018djw} can be considered closest to pseudo-Riemannian geometry. Their main feature is that the geometry defining length measure for curves is \textcolor{black}{typically} genuinely Finslerian (i.e., not pseudo-Riemannian), however, their geodesics are autoparallels of an affine connection on spacetime; in other words, the so-called canonical Finslerian nonlinear connection is equivalent to an affine connection on spacetime. These classes of Finsler spacetimes are particularly interesting as they are closely related to metric-affine geometry \cite{Pfeifer:2019tyy,Fuster:2020upk,Csillag:2026kdc}. Moreover, they have the advantage that the usually rather complicated Finsler gravity equation significantly simplifies to
\begin{align}
	D^{ab}R_{ab} = 0\,;
\end{align}
here, $D^{ab}$ is a tensorial object obtained from the Berwald-Finsler metric and the tangent space coordinates, and $R_{ab}$ are the components of the Ricci tensor of the canonical affine connection, discussed in detail in Section \ref{ssec:Berw}. In contrast to general relativity, Ricci flatness is not equivalent to solving the Finsler gravity equations -- it is only sufficient.

Thus, the following nontrivial question is immediate: \textit{What are all spherically symmetric, Berwald-Finsler spacetime vacuum solutions? Can they be classified, what are their properties?}

The starting point to answer these questions is the local classification of all spherically symmetric Berwald Finsler spacetimes, which is known from \cite{Cheraghchi:2022zgv,Voicu:2024mag}. It has been used to show  that the Jebsen-Birkhoff Theorem extends to Berwald Finsler spacetimes in the sense that $R_{ab}=0$ implies that spatially spherically symmetric Finsler spacetime geometry must be either completely flat, or pseudo-Riemannian Schwarzschild geometry \cite{Birkhoff-Th}.

\bigskip

\NV{In this article, we find (to the best of our knowledge, for the first time),  Berwald Finsler spacetime vacuum solutions that are \emph{neither Ricci flat, nor pseudo-Riemannian}; moreover, we find all such solutions that belong to the only class of spherically symmetric Finsler Berwald spacetimes \CP{that can admit flat geometries, i.e.\ can possess a well-defined flat spacetime limit.} These are candidates to describe the gravitational field of compact objects, beyond general relativity.}



To this aim, we proceed as follows.
\begin{itemize}
		\item Using the earlier found local classification of all spatially spherically symmetric (technically $SO\left( 3\right)$-symmetric), 4-dimensional \NV{non-Riemannian} Berwald-Finsler structures \cite{Cheraghchi:2022zgv}, we show  \NV{among the five existing classes, there is only one that allows for connections with vanishing curvature -- and hence, is the most appropriate for a discussion on flat spacetime limit.} 
		\item For this remaining class (deemed Class 3 in \cite{Voicu:2024mag}), we completely integrate the Finslerian gravity equation in vacuum. It turns out that there exist three families of solutions that are not Ricci-flat.
\end{itemize}

The presentation of the results throughout this article is organized as follows: In Section \ref{sec: BerFS} we introduce the necessary language to discuss all classes of spatially spherically symmetric Berwald-Finsler spacetimes. Next, in Section \ref{sec:class3} we identify the unique class (Class 3), among all spatially spherically symmetric Berwald Finsler spacetimes, that allows for flat geometries. \CP{Further, we discuss that they can be interpreted as $(\alpha,\beta)$-metrics and determine the form of their defining $1$-form $\mathbf{b}$ and pseudo-Riemannian metric $\mathbf{a}$. Finally, in Section \ref{sec:FEQC3} we solve the Finsler gravity vacuum equation explicitly and find two branches \NV{(containing one, respectively, two families)} of spherically symmetric, non-Ricci flat Berwald Finsler vacuum solutions of the Finsler gravity equation. The article is concluded in Section \ref{sec:CO}. Appendix \ref{app:defs} contains the connection and the nonlinear curvature components of spherically symmetric Berwald Finsler metrics and Appendix \ref{appx:A}, in which we display the general form of the pseudo-Riemannian metric $\mathbf{a}$ of Class 3 Finsler metrics.}

\section{Spatially spherically symmetric Berwald Finsler spacetimes}\label{sec: BerFS}

In this section, we briefly introduce the concepts of (pseudo)-Finsler space, Finsler spacetime and Finsler gravity. Moreover, we recall the definition of Berwald-Finsler spacetimes and the classification of all spatially spherically symmetric Berwald-Finsler spacetimes introduced in \cite{Cheraghchi:2022zgv}.  All of these concepts lay the foundation for the search of non-Ricci flat, causally well-behaved (see Sec. \ref{sec:class3}), vaccuum solutions in Section \ref{sec:FEQC3}.

We keep the the presentation here as short as possible and refer to  \cite{Cheraghchi:2022zgv,Birkhoff-Th,Voicu:2024mag} for further details.

\subsection{Pseudo-Finsler spaces, Finsler spacetimes and Finsler Gravity}

Consider a smooth $4$-dimensional manifold $M$ (\NV{interpreted as} spacetime) equipped with local coordinate charts~$\left( U,\left(x^{a}\right) \right)$, where indices $a,b,c,....$ run over $0,1,2,3$. Moreover, let $TM\equiv \left\{ \left( x,\dot{x}\right) ~|~x\in M,\dot{x}\in T_{x}M\right\}$, be the tangent bundle equipped with naturally induced local charts $\left( TU,\left( x^{a},\dot{x}^{a}\right)
\right)$,  that is, $\dot{x}=\dot{x}^{a}\partial _{a}$ is the decomposition of the tangent vector $\dot{x}\in T_{x}M$ in the natural local coordinate basis $\left\{ \partial _{a}\right\}$. If the chart is fixed and there is no risk of confusion, we will typically denote the coordinates $\left( x^{a},\dot{x}^{a}\right) $ simply as $\left( x,\dot{x}\right)$. Moreover, denote the natural local coordinate basis of $T_{(x,\dot{x})}TM$ by $\{\partial _{a}=\tfrac{\partial }{\partial x^{a}},\dot{\partial}_{a}=\tfrac{\partial }{\partial \dot{x}^{a}}\}$ and its dual, by $\{dx^a,d\dot x^a\}$. \NV{We will also denote for brevity, partial differentiation by the base coordinates $x^a$ with commas $_{,a}$ and  partial differentiation by the fiber coordinates $\dot{x}^a$ by dots $_{\cdot a}$.  }

\bigskip

A \textit{pseudo-Finsler structure} (or \textit{pseudo-Finsler metric}) on $M$, \cite{Bejancu}, is a smooth function $L:\mathcal{A}\rightarrow \mathbb{R}$ (sometimes also called a Finsler Lagrangian), defined on a conic subbundle\footnote{
A conic subbundle of $TM$ is defined as an open subset $\mathcal{A}\subset
TM\setminus \{0\}$ with non-empty fibers $\mathcal{A}_{x}=\mathcal{A}\cap
T_{x}M$ at all $x\in M,$ which is stable under positive rescaling of
vectors, i.e.: $\left( x,\dot{x}\right) \in \mathcal{A}\Rightarrow \left(
x,\lambda \dot{x}\right) \in \mathcal{A},\forall \lambda >0.$} 
$\mathcal{A}\subset TM\backslash \{0\},$ with the following properties:

\begin{enumerate}
\item Positive 2-homogeneity:\ $L\left( x,\lambda \dot{x}\right) =\lambda
^{2}L\left( x,\dot{x}\right) ,$ $\forall \lambda >0.$

\item Nondegeneracy: at any $\left( x,\dot{x}\right) \in \mathcal{A}$ and in
one (then, in any) local chart around $\left( x,\dot{x}\right) $, the
Hessian:%
\begin{equation}
g_{ab}\left( x,\dot{x}\right) 
:=\dfrac{1}{2} \dot{\partial}_a\dot{\partial}_b L
\end{equation}
is nonsingular.
\end{enumerate}
We note that any pseudo-Finsler function $L$ can be continuously prolonged
as $0$ at $\dot{x}=0$.

In the above, the conic bundle $\mathcal{A}$ is called the set of \textit{%
admissible vectors} and the functions $g_{ab}=g_{ab}\left( x,\dot{x}\right) $
are the local components of the \textit{Finslerian metric tensor} attached
to $L,$ which is a mapping 
\begin{equation}
g:\mathcal{A}\rightarrow T_{2}^{0}M,(x,\dot{x})\mapsto g_{(x,\dot{x}%
)}=g_{ab}dx^{a}\otimes dx^{b}.
\end{equation}
\textit{Terminology convention:} 
\textcolor{black}{In the following, we will typically omit the prefix `pseudo' and refer to pseudo-Finsler (respectively, pseudo-Riemannian) spacetimes simply as Finsler (respectively, Riemannian) spacetimes;} whenever needed, we will indicate the signature
explicitly.

The square root $F = \sqrt{|L|}$ defines a (parametrization-invariant) length measure for admissible curves $\gamma$ on $M$ through
\begin{align}\label{eq:Flength}
	\ell[\gamma] = \int F(\gamma(\tau),\dot \gamma(\tau))\, d\tau;
\end{align}
the curve $\gamma$ is called \textit{admissible} if its canonical lift   $(\gamma, \dot \gamma:=\frac{d\gamma}{d\tau})$ to $TM$ has its image completely contained in $\mathcal{A}$ .
Going from Finsler structures to Finsler spacetimes and Finsler gravity can be done in the following steps:
\begin{itemize}
	\item \textbf{Finsler spacetimes:} In the literature, there exist various
	definitions of Finsler spacetimes, e.g., \cite{Beem,Caponio:2015hca,Hohmann:2021zbt,Javaloyes:2018lex,Lammerzahl:2018lhw,Minguzzi:2014fxa}, which differ in nuances on the precise differentiability requirements imposed on $L$. Without entering the precise details of these definitions, they all tend to agree on the following: the manifold $M$ should be 4-dimensional, connected and orientable and the
	Finslerian metric tensor $g$ must be well-defined and have Lorentzian
	signature $\left( +,-,-,-\right) $ on an appropriate cone $\mathcal{T}_{x}$
	in each tangent space (on which $L>0$); moreover, on the boundary $\partial 
	\mathcal{T}_{x},$ one must have $L=0$. 
	
	The cones $\mathcal{T}_{x}$ are physically interpreted as \emph{future-pointing timelike cones}, \NV{containing} the physical $4$-velocities at which massive particles propagate through spacetime. The boundaries $\partial \mathcal{T}_{x}$, represent the light cones, interpreted as the \NV{sets of} directions along which massless particles propagate.
	
	\item \textbf{The Cartan tensor:} A first measure of the departure of a general
	Finsler function from a Riemannian one is given by the Cartan tensor  $%
	C=C_{abc}dx^{a}\otimes dx^{b}\otimes dx^{c},$ with local components $%
	C_{abc}=C_{abc}\left( x,\dot{x}\right) $ given by:
	\begin{equation}
		C_{abc}:=\dfrac{1}{2}\dot \partial_c\partial g_{ab}
		=\dfrac{1}{4}\dot \partial_c \dot \partial_b \dot \partial_a L\,.\label{def:C_a}
	\end{equation}%
	Its trace $C_{a}:=g^{bc}C_{abc}$ measures the $\dot{x}$-dependence of $\det g:$%
	\begin{equation}
		C_{a}=\dot \partial_a \sqrt{\left\vert \det g\right\vert }\,. \label{eq:C_a_det(g)}
	\end{equation}
	
	\item \textbf{Geodesics and canonical nonlinear connection:} Curves that extremize the length measure \eqref{eq:Flength} are called Finsler geodesics. In arc-length parametrization, they are given by solutions of the Finsler geodesic equation
	\begin{align}
		\ddot x^a + 2G^a(x,\dot x) = 0\,.
	\end{align}
	Here, $G^a$ are the so called geodesic spray components; in turn, they define the \emph{canonical Cartan nonlinear connection} coefficients $G^a{}_b$:
	\begin{align}\label{def:G^a_G^a_b}
		G^a = \frac{1}{4}g^{ab}(\dot x^c \partial_c \dot{\partial}_b L - \partial_b L)\,,\quad G^a{}_b = \dot{\partial}_b G^a\,.
	\end{align}
	The nonlinear connection coefficiens define the so-called \textit{horizontal derivative} operators
	\begin{align}
		\delta_a = \partial_a - G^b{}_a \dot \partial_b\,.
	\end{align}
		
	\item \textbf{Dynamical covariant derivative:} The canonical Cartan nonlinear connection leads to the definition of the \emph{dynamical covariant derivative}, that acts on on either horizontal \NV{vector fields  $X=X^{a}\delta _{a}$ on $\mathcal{A} \subset TM$ }, or on vertical ones $Y=Y^{a}\dot{\partial}_{a}$, as:
	\begin{equation}
		\nabla X=\left( \nabla X^{a}\right) \delta _{a}=\left( \dot{x}^{b}\delta
		_{b}X^{a}+G_{~b}^{a}X^{b}\right) \delta _{a},~\quad
		\nabla Y=\left( \nabla
		Y^a\right) \dot\partial _{a}=\left( \dot{x}^{b}\delta
		_{b}Y^{a}+G_{~b}^{a}Y^{b}\right) \dot{\partial}_{a}.
	\end{equation}%
	Intuitively, the dynamical covariant derivative measures the rate of change of the respective vector fields along (canonically lifted to $TM$) geodesics $\gamma$  \NV{of spacetime}. It is the canonical generalization of the Levi-Civita covariant derivative to Finsler geometry and also satisfies the properties:\ $\nabla L=0$, $\nabla g=0$.
	
	\item \textbf{Curvature and Finslerian Ricci scalar:} The \emph{curvature} of a nonlinear connection is, by definition
	the Lie bracket of the horizontal basis elements $\delta_a$; this is a tensor field $\mathcal{R}$ on $\mathcal{A} \subset TM$: 
	\begin{equation}
		\mathcal{R}=\frac{1}{2}R_{~bc}^{a}\dot{\partial}_{a}\otimes dx^{b}\wedge
		dx^{c},~\ \ \ \left[ \delta _{b},\delta _{c}\right] =:R_{~bc}^{a}\dot{%
			\partial}_{a}=(\delta _{c}G_{~b}^{a}-\delta _{b}G_{~c}^{a})\dot{\partial}%
		_{a}.\   \label{def:R}
	\end{equation}%
	It measures the non-integrability of the horizontal distribution $Span\{\delta_a\}$ and serves to characterize Finslerian geodesic deviation (see, e.g.,
	\cite{Miron}). The trace of the geodesic deviation operator is given by the $2$-homogeneous scalar
	\begin{equation}
		R=R_{~ac}^{a}\dot{x}^{c},  \label{def:Finsler_Ricci_scalar}
	\end{equation}%
	called the \emph{Finslerian Ricci scalar}.
	
	\item \textbf{The Landsberg tensor:} The Landsberg tensor $P = P_{abc}dx^a \otimes dx^b \otimes dx^c$, respectively, its trace $P_a dx^a$, 
	\begin{align}
		P_{abc} = \nabla C_{abc}\,,\quad P_a = g^{bc}P_{abc}\,.
	\end{align}
	measure the change of the Cartan tensor (respectively, of its trace) along geodesics on $M$.
	
	\item \textbf{Finsler gravity:} All the structures just introduced allow us to determine Finsler spacetimes dynamically from a physical gravitational field equation, similar to the way Einstein's equations determine the pseudo-Riemannian geometry of spacetime in general relativity.
	
	Using the ($0$-homogenized) Finslerian Ricci scalar $L^{-1}R$ as a field-theoretical
	Langrange function to construct an action on the tangent bundle for Lagrangians $L$, leads to the most promising Finslerian gravitational field equation, sourced by the 1-particle distribution function $\varphi$ of a kinetic gas \cite{Hohmann:2019sni}. The resulting field equation can be expressed as
	\begin{equation}
		g^{ab}\dot{\partial}_{a}\dot{\partial}_{b}R-6\dfrac{R}{L}+2 g^{ab}[2 \dot{\partial_a} (\nabla P_b) - \nabla (\dot{\partial_a} P_b)]=\kappa \varphi .  \label{eq:field_eqn}
	\end{equation}%
	Here $\kappa $ is the gravitational coupling constant. Recently, this equation was very successfully applied to determine the evolution of the universe \cite{Pfeifer:2025tda}. In homogeneous and isotropic symmetry, it becomes the so-called \textit{Finsler-Friedmann equation}, leading to an accelerated expansion of the universe.\CP{This offers a new prespective of explaining at least parts of the dark energy phenomenology by a change of spacetime geometry, instead of the introduction of additional fields}.
\end{itemize}

\textbf{Particular case: Riemannian geometry.} \NV{The case of (pseudo-)Riemannian geometry is singled out by the fact that $L$ is quadratic in the vector components $\dot{x}^a$:}
\begin{equation}
	L=a_{ab}\left( x\right) \dot{x}^{a}\dot{x}^{b}\,,
\end{equation}
where $a_{ab}(x)$ are the components of a pseudo-Riemannian metric; thus, 
\begin{equation}
\NV{g_{ab}(x,\dot{x}) = a_{ab}(x),}~ 2G^{a}\left( x,\dot{x}\right) = \Gamma _{~bc}^{a}\left( x\right) \dot{x}^{b}%
	\dot{x}^{c},~\ G_{~b}^{a}\left( x,\dot{x}\right) =\Gamma _{~bc}^{a}\left(
	x\right) \dot{x}^{c},
\end{equation}%
where $\Gamma _{~bc}^{a}$ are the Christoffel symbols of the metric $a$ and,
accordingly,%
\begin{equation}
	R_{~bc}^{a}(x,\dot x)=R_{d~bc}^{~a}(x)\dot{x}^{d},~\ 
	R(x,\dot x)=R_{d~ac}^{~a}(x)\dot{x}^{d}\dot{x}^{c}=-R_{cd}(x)\dot{x}^{c}\dot{x}^{d}\,.
\end{equation}
Here, $R_{d~bc}^{~a}$ and $R_{dc}=R_{d~ca}^{~a}$ represent the components of
the Riemannian curvature and of the Ricci tensor of $a.$ We note that the
Finsler-Ricci scalar of $a$ -- actually, the \emph{Raychaudhuri} \emph{scalar%
} of $a$ -- is different from the scalar curvature $a^{cd}R_{cd}$. Also, if
one sets $\varphi =0$ in (\ref{eq:field_eqn}) and restricts the search for $%
L $ to the class of Riemannian Finsler Lagrangians $L=a_{ab}\left( x\right) 
\dot{x}^{a}\dot{x}^{b},$ one obtains that our single equation (\ref%
{eq:field_eqn}) is equivalent to the set of 10 vacuum Einstein equations $%
R_{ab}=0.$ \\

Yet, in general, the functions $G_{~b}^{a}=G_{~b}^{a}\left( x,\dot{x}\right) 
$ are nonlinear in $\dot{x}$; the case when these turn out to be linear but $L\neq a_{ab}\left( x\right) \dot{x}^{a}\dot{x}^{b}$
deserves a special attention and will be discussed in the next subsection. We will see that, even in this most restrictive
Finslerian case, the field equation can have dramatically different
solutions from the Einstein equations.

\bigskip

\subsection{Finsler geometry of Berwald type}\label{ssec:Berw}
Berwald-Finsler structures are defined as (pseudo-)Finsler structures whose arc-length parametrized geodesics are autoparallel curves of a symmetric affine connection $\Gamma$ on the base manifold $M$; denoting its local coefficients by $\Gamma _{~ab}^{c}=\Gamma _{~ab}^{c}\left( x\right) $, this is
\begin{equation}
\frac{d^{2}x^{a}}{ds^{2}}+\Gamma _{~ab}^{c}(x)\dfrac{dx^{a}}{ds}\dfrac{dx^{b}%
}{ds}=0.  \label{eq:Berwald_geo}
\end{equation}%
In other words, a Finsler space $\left( M,L\right) $ is of Berwald type if
in any local chart, the canonical spray coefficients $G^{c}$ are quadratic, $G_{~b}^{c}$ are linear and $\dot{\partial}_{a}G_{~b}^{c}=%
\Gamma _{~bc}^{a}\left( x\right) $ are independent of $\dot{x}$: 
\begin{equation}
2G^{c}=\Gamma _{~ab}^{c}(x)\dot{x}^{a}\dot{x}^{b}~\ \Leftrightarrow
~G_{~b}^{c}=\Gamma _{~ab}^{c}(x)\dot{x}^{a}~\ \Leftrightarrow ~\dot{\partial}%
_{a}G_{~b}^{c}=\Gamma _{~ab}^{c}\left( x\right) .
\label{eq:coeffs_Berwald_spray}
\end{equation}
Thus, the dynamical covariant derivative $\nabla$ (on $\mathcal{A} \subset TM$) induces an affine connection $\Gamma$ on the base manifold $M$. 

As a consequence, the curvature of the canonical nonlinear connection and the Finslerian Ricci scalar of a Berwald space are given by
\begin{equation}\label{eq:nlcurvberw}
R_{~bc}^{a}(x,\dot x)=R_{d~bc}^{~a}\left( x\right) \dot{x}^{d},
\end{equation}%
where $R_{d~bc}^{~a}\left( x\right) $ are the curvature components of the induced affine connection on $M$. Moreover, one can show that for Berwald-Finsler structures the Landsberg tensor vanishes, see also~\cite{Bao},
\begin{align}
	P_{abc} = \nabla C_{abc} = 0\quad \Rightarrow\quad  P_a = \nabla C_{a}=0\,.
\end{align}

The above remarks on the curvature and the Lansberg tensor lead to a
considerable simplification of  the Finsler gravity equation \eqref{eq:field_eqn}. In vacuum ($\varphi =0$), it reads
\begin{equation}
g^{ab}\dot{\partial}_{a}\dot{\partial}_{b}R-6\dfrac{R}{L}=0 \quad \Leftrightarrow \quad  \left( g^{ab}(x,\dot x) - 3\frac{\dot x^a \dot x^b}{L(x,\dot x)} \right) R_{ab}(x)= 0.
\label{eq:vacuum_Berwald}
\end{equation}

\textbf{Remark (}\emph{Connection Ricci tensor vs. Finsler-Ricci scalar):}
The $\dot{x}$-derivatives of the Finslerian Ricci scalar $\dot{\partial}_{a}\dot{\partial}_{b} R(x,\dot x)$ are
related to the affine Ricci tensor components $R_{ab}(x)=R_{a~bc}^{~c}(x)$ as:
\begin{equation}
R(x,\dot x)=\CP{-}R_{ab}(x)\dot{x}^{a}\dot{x}^{b}\Rightarrow \dot{\partial}_{a}\dot{\partial}_{b} R=-2R_{ab}.
\label{eq:R_to_R_ab}
\end{equation}

It thus becomes obvious that $R_{ab} = 0$ is a sufficient condition (but not a necessary one, as we will prove below) for a Berwald-Finsler function $L$ to represent a vacuum solution of the Finsler gravity equation \eqref{eq:vacuum_Berwald}. As, in spatial spherical symmetry, this condition constrains a Berwald-Finsler spacetime to be either Riemannian (Schwarzschild), or flat 
(see \cite{Birkhoff-Th}), physically interesting vacuum solutions are those that do \textit{not} satisfy it.

\subsection{All spatially spherically symmetric Berwald metrics}\label{ssec:classif_Berwald}

To \emph{locally}\footnote{It might be possible to extend the locally characterized spacetimes to further coordinate patches by following geodesics -- as one does, e.g., for the Schwarzschild solution, by passing to Kruskal-Szekeres coordinates. We will not discuss such extensions here.} 
characterize all spatially spherically symmetric Finsler structures of Berwald type using coordinates, we can assume with no loss of generality that we are working on $M$ as an open subset of $\mathbb{R}^{4}\backslash \{0\}$, equipped with spherical coordinates $\left( t,r,\theta ,\phi \right)$. Moreover, we assume that $M$ has ends, i.e., it admits points with $r\rightarrow \infty$.

\bigskip

A pseudo-Finsler function $L:\mathcal{A\rightarrow }\mathbb{R}$ as above is called \textit{spatially spherically symmetric} if:

\begin{enumerate}
	\item $SO\left( 3\right) $ is a subgroup of the isometry group of $L$, i.e., $SO(3)$ acts by diffeomorphisms $\Psi:M \rightarrow M$ such that  $L(\Psi\left(x\right) ,d\Psi _{x}\left( \dot{x}\right)) = L\left( x,\dot{x}\right),$ for all $\left( x,\dot{x}\right)~\in~\mathcal{A}$.
	\item \NV{The orbits of $SO(3)$ are two-dimensional spheres}.
	\item The elements of $SO\left( 3\right) $ leave the coordinate $t$ (usually interpreted as time) invariant.
\end{enumerate}
In naturally induced coordinates $\left( x,\dot{x}\right) :=\left(
t,r,\theta ,\phi ,\dot{t},\dot{r},\dot{\theta},\dot{\phi}\right) $ on $TM,$
any spatially spherically symmetric Finsler metric is known to be, \cite{Pfeifer:2011xi}, of the form
\begin{equation*}
L(x,\dot{x})=L(t,r,\dot{t},\dot{r},w),\quad w^{2}=\dot{\theta}^{2}+\dot{\phi}%
^{2}\sin ^{2}\theta \,.
\end{equation*}

\bigskip

Nontrivially Finslerian (i.e., non-quadratic)\ Berwald metrics with spatial
spherical symmetry, as well as their canonical affine connections have been locally classified in \cite{Cheraghchi:2022zgv}. It was found that their coefficients $\Gamma^a{}_{bc}(x)$ of are encoded into 10 functions $k_{i}=k_{i}(t,r)$ and the curvature
tensor $R^a{}_{bcd}(x)$ of $\Gamma$ is described by 14 coefficients $a_{i}=a_{i}(t,r);$ see Appendix \ref{app:defs} for the precise definitions of $k_{i}$ and $a_{i}$ and for the
relations between them. Depending on which relations hold between $k_{i}$ and $a_{i} $, the Berwald condition
\begin{equation}\label{eq:bercond}
	\delta_a L = \partial_a L - \Gamma^b{}_{ac}(x)\dot{\partial_ b}L = 0\,,
\end{equation}
can be solved and one finds two major cases of spatially spherically symmetric Berwald Finsler structures, that contain several subclasses:

\begin{itemize}
\item[Case I.] \textbf{If }$k_{7},k_{8},k_{9},k_{10}$\textbf{\ are not all
zero}, one may assume with no loss of generality, that $k_{10}\not=0.$
Denoting
\begin{equation}
a:=\dfrac{k_{7}}{k_{10}},~\ \ b=\dfrac{k_{8}}{k_{10}},~\ \ c=\dfrac{%
k_{9}k_{10}-k_{7}k_{8}}{k_{10}^{2}},  \label{def:abc}
\end{equation}%
the connection must satisfy:%
\begin{equation} 
\begin{split}
A& :=b\left( aa_{1}+a_{2}\right) +\left( ab+c\right) \left(
aa_{3}+a_{4}\right) -a_{5}\left( 2ab+c\right) =0\,, \\
B& :=a\left( aa_{3}+a_{4}\right) -\left( aa_{1}+a_{2}\right) =0, \\
C& :=\left( ab+c\right) a_{3}+b\left( aa_{3}+a_{4}\right) +b(a_{1}-2a_{5})=0,
\\
a_{6}& =aa_{7},~a_{8}=ba_{7},~\ a_{9}=\left( ab+c\right) a_{7},~\
a_{10}=aa_{11},~a_{12}=ba_{11},~\ \ a_{13}=\left( ab+c\right) a_{11}.
\end{split}
\label{eq:ABC}
\end{equation}

Further, depending on the quantities%
\begin{eqnarray}
D &=&aa_{3}-a_{1}+a_{5},~\ \ E=ba_{3},~\ \ F=aa_{3}-a_{1},  \label{eq:DEF} \\
G &=&2\left( k_{1}-k_{4}a\right) ,~\ \ \tilde{G}=G-2k_{8},~\ \ \ H=2\left(
k_{2}-k_{6}a\right) ,~\ \ \ \tilde{H}=H-2k_{9},
\end{eqnarray}%
one obtains in this case, three classes:
\end{itemize}

\begin{enumerate}
\item[\textbf{Class 1:}] $D\neq 0$. Then, $\left[ \delta _{t},\delta _{r}%
\right] \not=0$ and $L$ has a \textit{power law }formula:%
\begin{equation}
L=\vartheta (t,r)u^{2-2\lambda }\left( v+\rho u^{2}\right) ^{\lambda }\,,
\label{def:power_law}
\end{equation}%
where:%
\begin{equation}
u=\dot{t}-a\dot{r},\quad v=c\dot{r}^{2}+2b\dot{t}\dot{r}-w^{2}\,,
\label{def:uv}
\end{equation}%
$\lambda =\frac{F}{D}\not=1$ is a constant\footnote{%
The value $\lambda =1$ corresponds to Riemannian metrics, which thus do not
satisfy the hypothesis of being properly Finslerian.}, $\rho =\frac{E}{D}$
and $\vartheta (t,r)=\exp \left( \int \left( G-\lambda \tilde{G}\right)dt+\left( H-\lambda \tilde{H}\right) dr\right) $ (where the integral is
taken along an arbitrary path in the $\left( t,r\right) $ plane, joining
some $\left( t_{0},r_{0}\right) $ to $\left( t,r\right) $).

\item[\textbf{Class 2:}] $D=0,\, E,F\neq 0$. Then $\left[ \delta _{t},\delta_{r}\right] \not=0$ and $L$ is of \textit{exponential type:}%
\begin{equation}
L=\varphi (t,r)u^{2}\exp (\frac{v}{u^{2}}\mu )\,,
\label{def_exponential_law}
\end{equation}%
where $\mu =\frac{F}{E}$, $u,v$ are as above and $\varphi (t,r)=\exp \left(
\int \left( G+2k_{4}b\mu \right) dt+\left( H+2k_{6}b\mu \right) dr\right) $
(where the integral is taken along an arbitrary path in the $\left(
t,r\right) $ plane, joining some $\left( t_{0},r_{0}\right) $ to $\left(
t,r\right) $).

\item[\textbf{Class 3:}] $D=E=F=0.$ Then, $\left[ \delta _{t},\delta _{r}%
\right] =0$ and:%
\begin{equation}
L=u^{2}\Xi (z),\quad z=\frac{v(\dot{\tilde{t}},\dot{\tilde{r}},w)}{u(\dot{%
\tilde{t}},\dot{\tilde{r}},w)^{2}}\,,  \label{def_class_3}
\end{equation}%
where $\Xi =\Xi (z)$ is an arbitrary function of the single variable $z=%
\dfrac{v}{u^{2}}$ and $u,v$ as defined above are, up to a suitable
coordinate change, independent of the new coordinates $\tilde{t}$ and $%
\tilde{r}$. In this case, $b$ and $c$ cannot simultaneously vanish.
\end{enumerate}

\begin{itemize}
\item[Case II.] \textbf{If} $k_{7}=k_{8}=k_{9}=k_{10}=0$, we distinguish two
more classes:
\end{itemize}

\begin{enumerate}
\item[\textbf{Class 4:}] $\left[ \delta _{t},\delta _{r}\right] =0.$ Then,
up to a suitable coordinate change $\left( t,r\right) \mapsto \left( \tilde{t%
},\tilde{r}\right) $, $L$ is given by a free function of two variables: 
\begin{equation}
L=L(\dot{t},\dot{r},w)=\dot{t}^{2}L(1,p,s)\,,\quad p=\frac{\dot{r}}{\dot{t}}%
,\quad s=\frac{w}{\dot{t}}\,.  \label{def_class_5}
\end{equation}

\item[\textbf{Class 5:}] $\left[ \delta _{t},\delta _{r}\right] \not=0$.
Then $a_{1}a_{4}-a_{2}a_{3}\not=0$ and%
\begin{equation}
L=w^{2}\xi (q)\,,\quad q=\frac{\dot{t}e^{I\left( p\right) -\varphi }}{w},~\
\ \ p=\frac{\dot{r}}{\dot{t}},  \label{def_class_4}
\end{equation}%
where $\xi =\xi (q)$ is an arbitrary function of the single variable $q;$ $%
\varphi $ is a function of $t,r$ only and%
\begin{equation}
I=\int \dfrac{a_{1}+a_{2}p}{a_{2}p^{2}+\left( a_{1}-a_{4}\right) p-a_{3}}dp.
\label{I_p}
\end{equation}
\end{enumerate}

\bigskip

\NV{An interesting question is which of these classes can contain flat Finsler spacetime structures.}
\NV{A \textit{flat} (also called in the literature, \textit{locally Minkowski}\footnote{\textcolor{black}{We prefer, though, the term 'flat' since, on the one hand, it is more suggestive for physical applications and, on the other hand, it avoids possible confusions with the Minkowski metric of special relativity.}} Finsler spacetime is defined as a Finsler spacetime $(M,L)$ with the property that $M$ admits an atlas composed by charts such that in the naturally induced coordinates  on $TM$, $L$ is expressible as $L=L(\dot{x})$ only. It is obvious that, for Berwald-Finsler spaces, the condition $L=L(\dot{x})$ implies that the canonical affine connection has identically vanishing curvature.}
\NV{In our notations, this is translated as:}
\begin{equation}
a_{i} = 0,~\ \ \ i=1,...,14.
\label{def:asympt_flatness}
\end{equation}

\NV{A quick inspection shows that the above condition implies, for Case I, that $A=B=C=D=E=F=0$, which selects Class 3 as the only possible subclass. Moreover, Classes 4 and 5 (i.e., Case II) are excluded as their defining condition $k_{7}=k_{8}=k_{9}=k_{10}=0,$ implies $a_{14}=1+k_7k_8+k_9k_{10} =1.$}

\NV{We conclude that the only class of spherically symmetric Berwald-Finsler spacetimes which also contains flat metrics is Class 3, making it the most interesting one in view of future studies on asymptotic flatness. In the following, we will study more in detail their structure, before determining all vacuum solutions of the Finsler gravity equation \eqref{eq:field_eqn} belonging to this class.}

\section{The structure of Class 3 metrics}\label{sec:class3}
Having identified Class 3 of Section \ref{ssec:classif_Berwald} as \CP{a most promsing} candidate for \CP{finding} non-Ricci flat solutions of the Finsler gravity equation \CP{with some flat spacetime limit}, \NV{we first} collect and recall several properties of these Berwald-Finsler structures -- especially, the fact that Finsler Lagrangians of Class 3 Berwald-Finsler structures as so-called $(\alpha,\beta)$-metrics $L=L(\mathbf{A},\mathbf{B})$, that are constructed from a pseudo-Riemannian metric factor $\mathbf{A} = a_{bc}(x)\dot x^b \dot x^c$ and a $1$-form factor $\mathbf{B} = b_c(x)\dot x^c$. \NV{Then, by choosing a convenient spherical coordinate system, we show that the curvature constraints obeyed by Class 3 actually specify $\mathbf{A}$ up to a free function of a single variable. Both these considerations will be of core importance in our study of the field equation, to be done in Section \ref{sec:FEQC3}.}

In the following, consider as input is an $SO(3)$ (spatially spherically) symmetric, torsion-free affine connection $\Gamma$ on \NV{the open subset $M \subset \mathbb{R}^4 \setminus \{0\}$}, with connection coefficients $\Gamma^a{}_{bc}(x)$ characterized by the nonvanishing functions $k_i(x), i=1,...,10$; we assume that the given connection satisfies the Finsler metrizability constraints \eqref{eq:bercond} and the Class 3 conditions, i.e\,,
\begin{equation}
A=B=C=D=E=F=0,  \label{eq:class_3_AF}
\end{equation}%
$k_{10}\not=0$ and $b,c\ $- not both zero. For the covariant derivative of objects on $M$ with respect to this affine connection we will use the semicolon notation, for example $b_{a;c}$.\\

\NV{As already stated in the previous section, the defining condition \eqref{eq:class_3_AF} for this class is equivalent to $[\delta_t,\delta_r]=0$, i.e., to }
\begin{equation} \label{eq:a1-a5_Class3}
	a_1 = a_2 = a_3 = a_4 = a_5 =0.
\end{equation}

\subsection{Class 3 metrics as $(\alpha,\beta)$-metrics}

\CP{Class 3 metrics can be interpreted as $(\alpha,\beta)$-metrics, which is very useful for the following derivations and to solve the Finsler gravity equation in Section \ref{sec:FEQC3}.} \NV{To perform this step, we start from an already known result, \cite{Voicu:2024mag}:}
\begin{lemma}\label{lem:Class_3} Assume $L$ belongs to Class 3, then, $L$ must be of the form:%
\begin{equation}
L=(e^{\mathcal{G}}u^{2})\Theta \left( ze^{-\left( \mathcal{G}-2\mathcal{K}%
\right) }+\mathcal{M}\right) ,  \label{general_sol_Berwald_class_3}
\end{equation}%
where $z:=\dfrac{v}{u^{2}},$ $\Theta $ is a free real function of a single
variable, $u=\dot{t}-a\dot{r}$, $v=c\dot{r}^{2}+2b\dot{t}\dot{r}-w^{2}$ and $\mathcal{G},\mathcal{K}$ and $\mathcal{M}$ are functions of $t$
and $r$ that solve the system:
\begin{eqnarray}
2\left( k_{1}-k_{4}a\right) &=&\partial _{t}\mathcal{G},~\ 2\left(
k_{2}-k_{6}a\right) =\partial _{r}\mathcal{G},~\ \ k_{8}=\partial _{t}%
\mathcal{K},~\ \ k_{9}=\partial _{r}\mathcal{K},  \label{def_G_K} \\
e^{-\left( \mathcal{G}-2\mathcal{K}\right) }bk_{6} &=&\dfrac{1}{2}\partial
_{r}\mathcal{M},~\ \ \ e^{-\left( \mathcal{G}-2\mathcal{K}\right) }bk_{4}=%
\dfrac{1}{2}\partial _{t}\mathcal{M}.  \label{def_M}
\end{eqnarray}
\end{lemma}
It was shown in \cite{Voicu:2024mag} that, if the given connection satisfies the
Finsler metrizability conditions, then the functions $\mathcal{G},\mathcal{K}
$ and $\mathcal{M}$ solving (\ref{def_G_K})-(\ref{def_M}) always exist.\\

In order to understand the structure of the Class 3 Berwald Finsler geometries, it is very convenient to introduce a $1$-form $\mathbf{b} = b_c(x)dx^c$ and a pseudo-Riemannian metric $\mathbf{a} = a_{bc}(x)dx^b \otimes dx^c$ on spacetime, as follows:
\begin{enumerate}
	\item Choosing $\Theta = 1$ in \eqref{general_sol_Berwald_class_3}, one finds $L=\mathbf{B}^2$, where 
	\begin{equation}
		\mathbf{B}=e^{\frac{\mathcal{G}}{2}}u~=~e^{\frac{\mathcal{G}}{2}}\left( \dot{t}-a\dot{r}\right)\,. \label{B}
	\end{equation}
	As $\mathbf{B}$ is linear in $\dot x$, it can be understood as $\mathbf{B}=b_c(x)\dot x^c$, induced by the $1$-form
	\begin{equation}
		\mathbf{b}=e^{\frac{\mathcal{G}}{2}} dt - ae^{\frac{\mathcal{G}}{2}} dr
	\end{equation}
	on $M$. Moreover, as  \eqref{def_G_K}-\eqref{def_M} guarantee that the Berwald condition $\delta_a L = \delta_a \mathbf{B}^2 = 0$ holds, this $1$-form is absolutely parallel with respect to the affine connection on $M$, i.e., it satisfies
	\begin{equation}\label{eq:para}
		\nabla_d b_c = 0\,.
	\end{equation}
	
	
	\item Alternatively,  choosing in \eqref{general_sol_Berwald_class_3}, $\Theta :=id.$, we find that $L$ is a quadratic function in $\dot{x}$:
	\begin{equation}
		L=\mathbf{A}:=e^{\mathcal{G}}\left( ve^{-\left( \mathcal{G}-2\mathcal{K}\right)}+\mathcal{M}u^{2}\right) =ve^{2\mathcal{K}}+\left( e^{\mathcal{G}}\mathcal{M}\right) u^{2}\,.  \label{A}
	\end{equation}
	The function $\mathbf{A}$ defines a pseudo-Riemannian metric
	\begin{equation}
		\mathbf{a} = a_{bc}(x)dx^b \otimes dx^c
	\end{equation}
	with components $a_{bc} = \textcolor{black}{\dfrac{1}{2}} \dot\partial_b \dot{\partial}_c\mathbf{A}$; this is defined and nondegenerate at least on a \textcolor{black}{(nonempty, open)} subset $D\subset M$ (which can be specified by appropriately choosing the range of the $(t, r)$ coordinates,
	see \cite{Voicu:2024mag}). As \eqref{def_G_K}-\eqref{def_M} guarantee that the Berwald condition $\delta_a L = \delta_a \mathbf{A} = 0$ holds, \textcolor{black}{the pseudo-Riemannian metric satisfies $a_{bc;a}=0$ for the covariant derivative defined through the given affine connection coefficients $\Gamma^a{}_{bc}(x)$ -- and thus $\Gamma^a{}_{bc}(x)$ must define the Levi-Civita connection of $\mathbf{a}$.}
\end{enumerate}

Using the newly defined quantities $\mathbf{A}$ and $\mathbf{B},$ we find that all Berwald Finsler structures of Class~3, \eqref{general_sol_Berwald_class_3}, can be expressed as:
\begin{equation}\label{eq:Lpsi}
L=\mathbf{B}^{2}\Theta \left( \dfrac{\mathbf{A}}{\mathbf{B}^{2}}\right)  = \mathbf{A}\Psi \left(s\right)\,,
\end{equation}
where 
\begin{align}
	s:=\dfrac{\mathbf{B}^{2}}{\mathbf{A}}, \quad \textrm{and} \Psi \left( s\right) :=s\Theta \left( \dfrac{1}{s}\right)\,.
\end{align}
We have thus proven:

\begin{lemma}\label{lem:alpha-beta_metric}
\begin{enumerate}
	\item All nondegenerate Class 3 solutions $L$ of the Berwald conditions are $%
	\left( \alpha ,\beta \right) $-metrics, with absolutely parallel 1-form $%
	\mathbf{b}$.
	\item The canonical affine connection of a Class 3 Berwald-Finsler function $
	L $ is the Levi-Civita connection of its pseudo-Riemannian constitutent $\mathbf{A}$.
\end{enumerate}
\end{lemma}
 This Lemma means that geodesics of the Finsler metric $L$ and those of $\mathbf{A}$ are the same as \textit{sets of points} - i.e., judging by their shapes only, we could not tell what is the correct model for our universe - the Finslerian or the Riemannian one. Yet, the proper time measurement along these worldlines is different when working with $L$, compared to $\mathbf{A}$.\\

Here are some immediate properties of Class 3 metrics, which we will need over the next sections.
\begin{lemma}
\label{lem:Ricci_b} For all Class 3 Finsler metrics:

\begin{enumerate}
\item There holds the identity%
\begin{equation}
R_{ab}\mathbf{\tilde{b}}^{b}=0,  \label{eq:Ricci_b}
\end{equation}%
where tildes denote raising indices by $a^{ab},$ that is: $\mathbf{\tilde{b}}%
^{b}=a^{bc}\mathbf{b}_{c}$.

\item The quantity $\left\langle \mathbf{b,b}\right\rangle
:=a^{ab}b_{a}b_{b} $ is a constant.

\item The curvature components $a_{7}$ and $a_{11}$ obey:%
\begin{equation}
\left( ab+c\right) a_{7}=ba_{11}.  \label{a7a11}
\end{equation}
\end{enumerate}
\end{lemma}

\begin{proof}
The first result follows immediately from the Ricci identities $\mathbf{%
\tilde{b}}_{~;a;d}^{c}-\mathbf{\tilde{b}}_{~;d;a}^{c}=R_{a~bd}^{~c}\mathbf{%
\tilde{b}}^{a}$ applied to the vector field $\mathbf{\tilde{b}}=\mathbf{%
\tilde{b}}^{a}\partial _{a}$ and to the fact that this vector field is
absolutely parallel with respect to $\Gamma.$ The second one is also
obtained as a direct consequence of the fact that the 1-form $\mathbf{b}$ is
covariantly constant and $\Gamma$ is the Levi-Civita connection of $a$.
Finally, statement 3. follows from
\begin{equation}
a_{5}\equiv a_{9}-a_{12}=\left( ab+c\right) a_{7}-ba_{11}=0.
\end{equation}
\end{proof}

\CP{In order to further simplify the steps to find solutions of the Finsler gravity equation, we next use the above mentioned Class 3 constraints to determine the form of the metric defining building blocks $\mathbf{A}$ and $\mathbf{B}$.}


\subsection{The form of $\mathbf{A}$ and $\mathbf{B}$ for class 3 metrics}\label{ssec:ABC3}
In the following, taking into account that the canonical affine connection of a Class 3 metric is actually, the Levi-Civita connection of its Riemannian part $\mathbf{A}$, we will answer the question: What are the Riemannian metrics $\mathbf{A}$, defined by \eqref{A}, whose Levi-Civita connections satisfy the Finsler metrizability conditions and the Class 3 constraints \eqref{eq:class_3_AF}?

\subsubsection{Finding a convenient coordinate system}
To this aim, a first step it to use the peculiar structure of $\mathbf{A}$ and $\mathbf{B}$ to find a convenient set of local
coordinates $\left( \tilde{t},\tilde{r},\theta ,\phi \right) $ in which the
expression of $\mathbf{A}$ becomes particularly simple.

We recall that
\begin{eqnarray}
	\mathbf{A} &=&ve^{2\mathcal{K}}+\left( e^{\mathcal{G}}\mathcal{M}\right)
	u^{2}=\mathcal{M}e^{\mathcal{G}}\dot{t}^{2}+2\left( be^{2\mathcal{K}}-%
	\mathcal{M}e^{\mathcal{G}}a\right) \dot{t}\dot{r}+\left( \mathcal{M}e^{%
		\mathcal{G}}a^{2}+ce^{2\mathcal{K}}\right) \dot{r}^{2}-e^{2\mathcal{K}}w^{2}\label{eq:Aexp}\\
	\mathbf{B} &=&e^{\tfrac{\mathcal{G}}{2}}u=e^{\tfrac{\mathcal{G}}{2}}\left( 
	\dot{t}-a\dot{r}\right)\label{eq:Bexp}\,,
\end{eqnarray}%
where in the second equality, we have used $u=\dot{t}-a\dot{r},~\ \ v=c\dot{r}^{2}+2b\dot{t}\dot{r}-w^{2}$.

The property that the 1-form $\mathbf{b}$ is absolutely parallel, $\nabla
_{a}b_{c}=0,$ entails in particular, that $b_{c,a}-b_{a,c}=0,$ that is its differential vanishes 
\begin{equation}
	d\mathbf{b=0.}
\end{equation}%
In turns, this entails that $\mathbf{b}$ is locally exact, meaning that, at
least locally, there exists a function $\tilde{t}=\tilde{t}\left( t,r\right)
,$ such that%
\begin{equation}
	\mathbf{b}=e^{\tfrac{\mathcal{G}}{2}}dt-ae^{\tfrac{\mathcal{G}}{2}}dr=d%
	\tilde{t}.
\end{equation}%
In particular:%
\begin{equation}
	\dfrac{\partial \tilde{t}}{\partial t}=b_{,t}=e^{\tfrac{\mathcal{G}}{2}},~\ 
	\dfrac{\partial \tilde{t}}{\partial r}=b_{,r}=-ae^{\tfrac{\mathcal{G}}{2}}.\ 
\end{equation}%
Furthermore, setting:%
\begin{equation}
	\tilde{r}:=e^{\mathcal{K}},
\end{equation}%
and taking into account that $b=\dfrac{k_{8}}{k_{10}},$ $ab+c=\dfrac{k_{9}}{%
	k_{10}},$ we have:%
\begin{eqnarray}
	\dfrac{\partial \tilde{r}}{\partial t} &=&e^{\mathcal{K}}\mathcal{K}_{,t}=e^{%
		\mathcal{K}}k_{8}=e^{\mathcal{K}}bk_{10}, \\
	\dfrac{\partial \tilde{r}}{\partial r} &=&e^{\mathcal{K}}\mathcal{K}_{,r}=e^{%
		\mathcal{K}}\left( ab+c\right) k_{10}
\end{eqnarray}%
and hence the Jacobian determinant of the mapping $\left( t,r,\theta
,\varphi \right) \mapsto \left( \tilde{t}\left( t,r\right) ,\tilde{r}\left(
t,r\right) ,\theta ,\varphi \right) $ is:%
\begin{equation} \label{Jacobian}
	\dfrac{D\left( \tilde{t},\tilde{r},\theta ,\phi \right) }{D\left( t,r,\theta
		,\phi \right) }=e^{\frac{\mathcal{G}}{2}+\mathcal{K}}\left( 2ab+c\right)
	k_{10}.
\end{equation}

\textbf{Assumption.} Let us suppose, for the rest of this section, that
\begin{equation}\label{eq:2ab+cnon0}
	2ab+c \neq 0\,.
\end{equation}

Using  expression \eqref{eq:ainv} of $a^{ab}$, we find by direct computation that this is equivalent to the fact that the constant norm of $b$ is non-zero (the factor $\Delta$ steems from the determinant of $\mathbf{a}$ and is derived also in Appendix \ref{appx:A}),
\begin{align}
	\left\langle \mathbf{b,b}\right\rangle =\dfrac{e^{\mathcal{G}}\left(
		2ab+c\right) }{\Delta }\not=0\,,\quad  
	\Delta :=\mathcal{M}e^{\mathcal{G}}\left( 2ab+c\right) -b^{2}e^{2\mathcal{K}}\,.
\end{align}%

\textbf{Remark.} Actually, later on, when discussing the field equation in Section \ref{sec:FEQC3}, we will find see that \eqref{eq:2ab+cnon0} (equivalently, $\langle \mathbf{b},\mathbf{b} \rangle \neq 0 $) is always entailed by the field equation, hence there will be no loss of generality in assuming it.

Under this assumption, and taking into account that $k_{10}\neq 0$, the Jacobian determinant \eqref{Jacobian} is nowhere zero; thus, we can
use locally (on any star-shaped coordinate chart), the coordinates
\thinspace $\left( \tilde{t},\tilde{r},\theta ,\phi \right) .$ This induces
a transformation of the fiber coordinates on $TM$ as:%
\begin{equation}
	\overset{\cdot }{\tilde{t}}~=\mathbf{B}~=e^{\tfrac{\mathcal{G}}{2}}\left( 
	\dot{t}-a\dot{r}\right) ,~\overset{\cdot }{\tilde{r}}~=e^{\mathcal{K}%
	}k_{10}\left( b\dot{t}+\left( ab+c\right) \dot{r}\right) ,~\ \overset{\cdot }%
	{\tilde{\theta}}~=\dot{\theta},~\ \ \overset{\cdot }{\tilde{\phi}}~=\dot{\phi%
	},\ \ 
\end{equation}%
with inverse 
\begin{equation}\label{eq:coord}
	\dot{t}=\dfrac{e^{\mathcal{K}}k_{10}\left( c+ab\right) \overset{\cdot }{%
			\tilde{t}}+ae^{\frac{\mathcal{G}}{2}}\overset{\cdot }{\tilde{r}}}{e^{\tfrac{%
				\mathcal{G}}{2}+\mathcal{K}}k_{10}\left( c+2ab\right) },~\ \ \ \dot{r}=%
	\dfrac{-bk_{10}e^{\mathcal{K}}\overset{\cdot }{\tilde{t}}+e^{\frac{\mathcal{G%
			}}{2}}\overset{\cdot }{\tilde{r}}}{e^{\tfrac{\mathcal{G}}{2}+\mathcal{K}%
		}k_{10}\left( c+2ab\right) }.
\end{equation}

\bigskip

Expressing $\mathbf{A}$ in the new coordinates and taking into account that $\mathcal{M}-\dfrac{b^{2}e^{2\mathcal{K}-%
		\mathcal{G}}}{c+2ab}$ is actually, just $\dfrac{1}{\left\langle \mathbf{b,b}%
	\right\rangle }$ (and is therefore, a constant), we find:
\begin{proposition}\label{prop:A_n}
	If $\langle \mathbf{b},\mathbf{b} \rangle \neq 0$, then there locally exist on $M$ spherical coordinates $(\tilde{t},\tilde{r},\theta,\phi)$, such that in the naturally induced coordinates on $TM$, $\mathbf{A}$  has the form:
	\begin{equation}
		\mathbf{A}=\dfrac{1}{\left\langle \mathbf{b,b}\right\rangle }\overset{\cdot }%
		{\tilde{t}}\overset{}{^{2}}+\dfrac{1}{k_{10}^{2}\left( c+2ab\right) }\overset%
		{\cdot }{\tilde{r}}\overset{}{^{2}}-~\tilde{r}^{2}w^{2},
	\end{equation}%
	
\end{proposition}

\textbf{Remark: }Using again the identity $\left\langle \mathbf{b,b}%
\right\rangle =\dfrac{e^{\mathcal{G}}\left( 2ab+c\right) }{\Delta },$
together with the observation that, in Lorentzian signature, $\Delta <0,$ we find that:

\begin{enumerate}
	\item If $\left\langle \mathbf{b,b}\right\rangle >0,$ then $2ab+c<0$ and we
	can interpret $\tilde{t},\tilde{r}$ as the new time, respectively, radial
	coordinates.
	
	\item If $\left\langle \mathbf{b,b}\right\rangle <0,$ accordingly, $2ab+c>0,$
	then the physical interpretations of $\tilde{t}$ and $\tilde{r}$ are swapped.
\end{enumerate}

As mentioned above, the case that $2 ab + c =0$, is not of interest for our upcoming inverstigations, which is why we do not discuss it here further.

\subsubsection{Solving the curvature constraints for Class 3}
In the following, assuming $2ab+c\neq0$, we will use the more convenient coordinate system found above, to solve the Finsler metrizability conditions \eqref{eq:bercond} and the Class 3 constraints \eqref{eq:class_3_AF}. To simplify the writing, let us omit the tildes and re-denote $\tilde{t}$ as $t$
and $\tilde{r}$ as $r.$ Thus, in the above defined coordinates,
\begin{equation}
	\mathbf{A}=\dfrac{1}{\left\langle \mathbf{b,b}\right\rangle }\dot{t}^{2}+%
	\dfrac{1}{k_{10}^{2}\left( c+2ab\right) }\dot{r}^{2}-r^{2}w^{2}=:m\dot{t}%
	^{2}+n\dot{r}^{2}-r^{2}w^{2},~\ \ \mathbf{B}=\dot{t},
	\label{sol:A_B_Branch1}
\end{equation}%
where $m:=\dfrac{1}{\left\langle \mathbf{b,b}\right\rangle }$ is a constant
and $n:=\dfrac{1}{k_{10}^{2}\left( c+2ab\right) }$ is a smooth function of $t
$ and $r.$ A direct computation shows that its Levi-Civita connection is (as
expected) of the form \eqref{eq:appcon}, with:
\begin{eqnarray}
	\Gamma _{~rr}^{t} &=&k_{3}=-\dfrac{1}{2}\dfrac{n_{,t}}{n}~\ \ \Gamma
	_{~tr}^{r}=k_{6}=\dfrac{1}{2}\dfrac{n_{,t}}{n},~\ \Gamma _{~rr}^{r}=k_{5}=%
	\dfrac{1}{2}\dfrac{n_{,r}}{n},~ \\
	\Gamma _{\theta \theta }^{r} &=&\dfrac{\Gamma _{\phi \phi }^{r}}{\sin
		^{2}\theta }=k_{10}=\dfrac{r}{n},~\Gamma _{~r\theta }^{\theta }= \Gamma
	_{~r\phi }^{\phi }=k_{9}=\dfrac{1}{r},
\end{eqnarray}%
whereas $k_{1}=k_{2}=k_{4}=k_{7}=k_{8}=0;$\ in particular,
\begin{equation}\label{eq:abc}
	a=b=0,~\ \ \  c= \frac{k_9}{k_{10}} = \dfrac{n}{r^{2}}\,,
\end{equation}%
and thus $2ab+c = c \neq 0$, i.e.\ the assumed condition \eqref{eq:2ab+cnon0} is satisfied, as long as $n\neq0$.

Accordingly, the surviving curvature components, see \eqref{eq:a_i}, are given by:%
\begin{eqnarray}
	a_{2} &=&\dfrac{1}{2}\dfrac{n_{,tt}}{m}-\dfrac{1}{4}\dfrac{n_{,t}^{2}}{mn}%
	,~\ \ a_{3}=-\dfrac{1}{4}\dfrac{2n_{,tt}n-n_{,t}^{2}}{n^{2}} \\
	a_{7} &=&\dfrac{1}{2}\dfrac{rn_{,t}}{n^{2}},~\ \ a_{9}=\dfrac{1}{2}\dfrac{%
		n_{,t}}{nr},~ \\
	\ a_{10} &=&\dfrac{1}{2}\dfrac{n_{,t}r}{mn},~\ a_{11}=\dfrac{1}{2}\dfrac{%
		rn_{,r}}{n^{2}},~\ \ a_{12}=\dfrac{1}{2}\dfrac{n_{,t}}{nr},~a_{13}=\dfrac{1}{%
		2}\dfrac{n_{,r}}{nr},~\ \ a_{14}=1+\dfrac{1}{n}.
\end{eqnarray}

A quick look at the above formulas shows that the Finsler metrizability
conditions $a_{6}=aa_{7}$ then entail:%
\begin{equation}
	n_{,t}=0\Rightarrow n=n\left( r\right) .
\end{equation}%
Thus, all Class 3-metrics are independent of time, in the chosen coordinates, thus posses a timelike Killing vector field, or, in other words, are \emph{stationary}.

Calculating now all non-vanishing curvature components yields
\begin{equation}
	a_{11}=\dfrac{1}{2}\dfrac{rn_{,r}}{n^{2}},~\ \ a_{13}=\dfrac{1}{2}\dfrac{%
		n_{,r}}{nr},~\ \ a_{14}=1+\dfrac{1}{n}\,,  \label{eq:nonzero_curvature_comps}
\end{equation}%
all others vanish. Hence, all the other metrizability conditions, together with the
Class 3 constraints, i.e.,
\begin{eqnarray}
	A &=&B=C=D=E=F=0, \\
	a_{6} &=&aa_{7},~\ a_{8}=ba_{7},~\ \ a_{9}=\left( ab+c\right) a_{7} \\
	a_{10} &=&aa_{11},~\ a_{12}=ba_{11},~\ \ a_{13}=\left( ab+c\right) a_{11}.
\end{eqnarray}%
are identically satisfied.

Further, using $n=n\left( r\right) $ only, to calculate the Ricci tensor of
the connection (which, we recall, is the Levi-Civita connection of $\mathbf{A%
}$), we find that its only nonzero components are%
\begin{equation}
	R_{rr}=\dfrac{n_{,r}}{rn},~\ \ R_{\theta \theta }=\dfrac{1}{2}\dfrac{%
		2n^{2}+2n-rn_{,r}}{n^{2}}=\dfrac{R_{\phi \phi }}{\sin ^{2}\theta };
\end{equation}%
therefore, the Finsler-Ricci scalar is expressed as%
\begin{equation}
	R=-R_{ab}\dot{x}^{a}\dot{x}^{b}=-\dfrac{n_{,r}}{nr}\dot{r}^{2}+\dfrac{1}{2}%
	\dfrac{rn_{,r}-2n-2n^{2}}{n^{2}}w^{2}  \label{eq:Finsler_Ricci_n}
\end{equation}%
and thus 
\begin{equation}
	a^{ab}R_{\cdot ab}=\dfrac{1}{n}R_{\cdot rr}-\dfrac{1}{r^{2}}R_{\cdot \theta
		\theta }-\dfrac{1}{r^{2}\sin ^{2}\theta }R_{\cdot \phi \phi }=\dfrac{4}{%
		n^{2}r^{2}}\left( n^{2}+n-rn^{\prime }\right) .
	\label{eq:scalar_curvature_n}
\end{equation}

These expressions will be used in the field equation, to be discussed in the
next section.

\NV{Let us quickly summarize what we found so far. Berwald-Finsler spacetime structures belonging to Class 3 are all given by \eqref{general_sol_Berwald_class_3}, in particular, they are $(\alpha,\beta)$-metrics. Their canonical nonlinear connection induces an affine connection $\Gamma$ on spacetime that is the Levi-Civita connection of the metric $\mathbf{a}$, defined through \eqref{A} and the $1$-form $\mathbf{b}$ employed to define the Finsler Lagrangian \eqref{B} is absolutely parallel with respect to this connection. Even more, assuming $2ab+c\neq0$ (equivalently  $\langle \mathbf{b},\mathbf{b}\rangle\neq0$), we can choose the spherical coordinates $(t,r,\theta,\phi)$ such that the expression of $\mathbf{A} = a_{ab}\dot{x}^a\dot{x}^b $ and $\mathbf{B} = b_a \dot{x}^a$ simplify to \eqref{sol:A_B_Branch1}, with $a_{tt}=m = \dfrac{1}{\langle \mathbf{b},\mathbf{b} \rangle}$ being a constant and $a_{rr} = n(r)$ being a function of $r$ only. The question we answer next is which kind of Finsler Lagrangians $L$ of this type are non-Ricci flat and solve the Finsler gravity vacuum equations?}

\section{Solving the field equation for Class 3 Berwald-Finsler metrics}\label{sec:FEQC3}
Following the preparation in the previous sections, we can now solve the Finsler gravity equation for spherically symmetric Berwald Finsler Lagrangian belonging to Class 3. We find two branches, containing one, respectively two families of solutions that are displayed in \eqref{sol:L_Branch1} and \eqref{gen_sol_Branch_2}.

Let us rewrite the Finsler gravity vacuum equation \eqref{eq:vacuum_Berwald} for Berwald structures as
\begin{equation}
g^{ab}R_{\cdot ab}=6\dfrac{R}{L},  \label{field_eq}
\end{equation}%
taking into account the expression \eqref{eq:Lpsi} of Class 3 metrics. 

The fact that our metrics are of $\left( \alpha ,\beta \right) $-type allows us to explicitly express the inverse metric components $g^{ab},$ as:
\begin{equation}
	\begin{split}
	g^{ac}=\dfrac{1}{\Psi -s\Psi ^{\prime }}a^{ac}+\dfrac{\NV{-2}s^{2}\Psi ^{\prime
			\prime }\left( \Psi +\left\langle \mathbf{b,b}\right\rangle \Psi ^{\prime
		}-s\Psi ^{\prime }\right) }{\left( \Psi -s\Psi ^{\prime }\right) \mu A}\dot{x%
	}^{a}\dot{x}^{c} \\ +\dfrac{2Bs\Psi \Psi ^{\prime \prime }}{\left( \Psi -s\Psi
		^{\prime }\right) \mu A}\left( \dot{x}^{a}b^{c}+\dot{x}^{c}b^{a}\right) \NV{-}
	\dfrac{\Psi \left( \Psi ^{\prime }-s\Psi ^{\prime }\right) -2s\Psi \Psi
		^{\prime \prime }}{\left( \Psi -s\Psi ^{\prime }\right) \mu }b^{a}b^{c}
		\end{split}
\end{equation}
where $\nu =\left( \Psi -s\Psi ^{\prime }\right) \mu $ and $\mu $ is given
by:%
\begin{equation} \label{def:mu}
	\mu =\Psi \left( \Psi -s\Psi ^{\prime }\right) +\left( \left\langle \mathbf{%
		b,b}\right\rangle -s\right) \left( \Psi \Psi ^{\prime }+2s\Psi \Psi ^{\prime
		\prime }-s\left( \Psi ^{\prime }\right) ^{2}\right)\,.
\end{equation}%
\CP{Note that this expression for the inverse is an Erratum to the one presented in (A10)  of \cite{alpha-beta-metrics}. It differs by a factor of $-2$ in front of the $\dot x^a \dot x^c$ term and a factor of $-1$ in front of the $b^a b^c$ term. It can straighforwardly be checked with (6) in \cite{alpha-beta-metrics} that the expression presented here is correct.}\\
Moreover, we have taken $\Psi -s\Psi ^{\prime }$ in front as a common factor and re-denoted $\nu =\left( \Psi -s\Psi ^{\prime }\right) \mu $. \NV{Finally, we note that the quantities $\Psi - s \Psi'$ and $\mu$ appearing in the denominator of $g^{ab}$ must always be nonzero, for any nondegenerate Finsler metric of $(\alpha,\beta)$ type.}

Further, using Lemma (\ref{lem:Ricci_b}), we find that the terms in the
second line above vanish when contracted by $R_{ab},$ whereas $R_{\cdot ab}%
\dot{x}^{a}\dot{x}^{b}=2R.$ That is:%
\begin{equation}
	g^{ab}R_{\cdot ab}=\dfrac{1}{\Psi -s\Psi ^{\prime }}\left( \left(
	a^{ac}R_{\cdot ac}\right) +\dfrac{-2s^{2}\Psi ^{\prime \prime }\left( \Psi
		+\left\langle \mathbf{b,b}\right\rangle \Psi ^{\prime }-s\Psi ^{\prime
		}\right) }{\mu A}2R\right) .  \label{g_R}
\end{equation}

\textbf{Remark. }In relation (\ref{g_R}), the term $a^{ab}R_{\cdot ab}$ is, up to a factor of -2, nothing but the curvature scalar \textcolor{black}{$a^{ab}R_{ab}$} of the pseudo-Riemannian metric $a$. \vspace{10pt}

Substituting the above identity together with $L=A\Psi $ into equation (\ref%
{field_eq}) and taking into account that we have excluded the situation when
the Finsler-Ricci scalar $R$ vanishes, a direct computation leads to:

\begin{proposition}
For Class 3 metrics with $R\not=0,$, the
Berwald vacuum equation (\ref{field_eq}) can be re-expressed as:
\begin{equation}
2\Phi =\dfrac{\mathbf{A}\left( a^{ab}R_{\cdot ab}\right) }{R},
\label{Phi_ratio_R}
\end{equation}%
where:%
\begin{equation}
	\Phi :=\left( 3\dfrac{\Psi -s\Psi ^{\prime }}{\Psi }+2\dfrac{s^{2}}{\mu }%
	\Psi ^{\prime \prime }\left( \Psi +(\left\langle \mathbf{b,b}\right\rangle
	-s)\Psi ^{\prime }\right) \right)  \label{def_Phi}\,,
\end{equation}%
and $\mu =\mu \left( s\right) $ is given by (\ref{def:mu}). 
\end{proposition}

Expressing the Finsler gravity equation as in \eqref{Phi_ratio_R} has the advantage that the $\dot x$-dependence of the right hand side is manifestly the ratio between quadratic expressions given by $\mathbf{A}$, \eqref{A}, and $R$, \eqref{eq:nlcurvberw}.

In the following, we will integrate equation (\ref{Phi_ratio_R}). For a clear discussion, we will split our analysis into two cases, depending on whether the \textcolor{black}{scalar curvature $a^{ab}R_{ab}$ of $\textbf{a}$ (or, equivalently, its multiplied by $-2$ version $a^{ab}R_{\cdot ab}$)}   is nonzero, in Section \ref{ssec:Branch1} or zero, in Section \ref{ssec:Branch2}.

\subsection{Branch 1: \texorpdfstring{$a^{ab}R_{\cdot ab}\neq0$}{R1}}\label{ssec:Branch1}
\NV{In this subsection, we first use the peculiar form of the field equation \eqref{Phi_ratio_R} to reveal one more curvature constraint - which will lead to a precise expression for $\mathbf{A}$ and hence, also for the Finsler-Ricci scalar $R$;  the latter will then allow us to rewrite \eqref{Phi_ratio_R} as a second order ODE in $\Psi =\Psi (s)$ and integrate it.}

\subsubsection{Determination of $\mathbf{A}$ and of $R$}
We start by rewriting \eqref{Phi_ratio_R} as
\begin{align}
	\frac{\mathbf{A}}{R} =2 \Phi (a^{ab}R_{\cdot ab})^{-1}\,.
\end{align}
and note that, in $a^{ab}R_{\cdot ab}$ only depends on $t$ and $r$, whereas $\Phi $ only depends on $s=\tfrac{\mathbf{B}^{2}}{\mathbf{A}}$.  \NV{It thus follows that the ratio $\tfrac{\mathbf{A}}{R}$ can depend on $t,r$ and $s$ only; in particular, its $\dot{x}$-dependence must be completely encoded in $\mathbf{A}$ and $\mathbf{B}$. Since, on the other hand, $R$ is quadratic in $\dot{x}^{a}$, it follows that it must be of the form}
\begin{equation}
	R= \delta \mathbf{A}+\gamma \mathbf{B}^{2},  \label{R_A_B}
\end{equation}%
where $\delta = \delta(t,r)$ and $\gamma = \gamma(t,r)$ are (for the moment, undetermined) smooth functions.\vspace{11pt}

\NV{Having found the general form of $R$ from the field equation, we will now prove that this leads to the fact that
\begin{equation}
	2ab+c\neq0\,,
\end{equation}
or equivalently
\begin{align}
	\langle \mathbf{b},\mathbf{b} \rangle \neq 0\,,
\end{align}
(as announced in section \ref{ssec:ABC3}), which allows us to use the convenient coordinates \eqref{eq:coord} and thus the diagonalized expression of  $\mathbf{A}$ presented in Proposition~\ref{prop:A_n}.}

\textcolor{black}{
Let us assume, by contraposition, that $\langle \mathbf{b},\mathbf{b} \rangle =0$, i.e., $2ab+c=0$. Substituting $c=-2ab$ into \eqref{a7a11} then leads to $-aba_{7}=ba_{11}.$ Since, under the above assumption, $b=0$ is not possible
(as it would lead to $c=0$ and therefore, to a violation of the conditions that define Class 3 metrics, see the definition around \eqref{def_class_3}), we deduce that
\begin{equation}\label{eq:a7a11_contraposition}
	a_{11}=-aa_{7}.
\end{equation}%
Further, using $a_1 =...= a_5 =0$ and the expressions in Appendix \ref{app:defs} of the curvature components, we get
\begin{flalign}
	&	R = R^{a}_{ac}\dot{x}^c
	=\left( -a_{3}-2a_{8}\right) \dot{t}^{2}+\left(	a_{1}-a_{4}-2a_{9}-2a_{12}\right) \dot{t}\dot{r}+\left( a_{2}-2a_{13}\right) 
	\dot{r}^{2}+\left( a_{6}+a_{11}-a_{14}\right) w^{2}\\
   &  =-2ba_{7}\dot{t}^{2}-4ba_{11}\dot{t}\dot{r}-2\left( ab+c\right) a_{11}%
	\dot{r}^{2}+\left( aa_{7}+a_{11}-a_{14}\right) w^{2},
\end{flalign}
which, using \eqref{eq:a7a11_contraposition}, together with
 the identity $a_{14}=1+k_{10}^{2}\left( 2ab+c\right) $ (see Appendix \ref{app:defs}) and $\mathbf{B}=e^{\tfrac{\mathcal{G}}{2}}\left( \dot{t}-a\dot{r}\right) $, gives
\begin{equation}
	R=-2ba_{7}\left( \dot{t}-a\dot{r}\right) ^{2}-w^{2}=-2be^{-\mathcal{G}}a_{7}%
	\mathbf{B}^{2}-w^{2}.
\end{equation}%
Since, as we have seen above, $R$ must be a linear combination (with $\left(
t,r\right) $-dependent coefficients) of $\mathbf{A}$ and $\mathbf{B}^{2}%
\mathbf{,}$ it turns out from the above that $w^{2}$ should be a linear
combination of $\mathbf{A}$ and $\mathbf{B}^{2}$, say, $w^2 = \lambda(t,r) \mathbf{A} + \sigma(t,r) \mathbf{B}^2 $. Using the concrete expressions \eqref{eq:Aexp} and \eqref{eq:Bexp} of $\mathbf{A}$ and $\mathbf{B}$ and identifying the coefficients of $\dot{t}^2, \dot{t}\dot{r}, \dot{r}^2$ and $w^2$, a quick inspection shows that the resulting system for $\lambda$ and $\sigma$ has no solutions, unless $b=a=0$; but this would also imply then $c = -2ab$ = 0. But, according to the definition of Class 3 metric, $b$ and $c$ cannot simultaneously vanish, see the comments below \eqref{def_class_3}, or Theorem 4 point 2.ii in \cite{Cheraghchi:2022zgv}. Therefore we conclude that $2ab+c\neq0$ and $\left\langle \mathbf{b,b}\right\rangle$ must be, indeed, a nonzero constant.}

\bigskip

\NV{Having ensured that $2ab+c\neq0$, we can now employ the simplifying coordinate transformation and substitute} the expressions (\ref{sol:A_B_Branch1}) of $\mathbf{A}$ and $\mathbf{B}$ into \eqref{R_A_B}, which then leads to
\begin{equation}
	R=\left( \delta m+\gamma \right) \dot{t}^{2}+\delta \left( n\dot{r}^{2}-r^{2}w^{2}\right) .
\end{equation}%
On the other hand, we have seen above that for Class 3 metrics,
the Finsler-Ricci scalar $R\ $is given by (\ref{eq:Finsler_Ricci_n}).
Comparing the two expressions gives
\begin{equation}
	\gamma =-\delta m,~\ \ \delta n=-\dfrac{n_{,r}}{nr},~-\beta r^{2}=\dfrac{1}{2}%
	\dfrac{rn_{,r}-2n-2n^{2}}{n^{2}}.  \label{eqs:beta_gamma}
\end{equation}%
Eliminating $\delta $ from the last two equations, we get an ODE in $%
n=n\left( r\right) :$

\begin{equation}
	n_{,r}=-\dfrac{2}{r}\left( n+n^{2}\right) ,
\end{equation}%
which has the general solution%
\begin{equation}
	n=-\dfrac{1}{1-\kappa_1 r^{2}},
\end{equation}%
where $\alpha \in \mathbb{R}$ is an arbitrary constant, parametrizing the family of solutions. Substitution into
the second equation (\ref{eqs:beta_gamma}) then leads to $\delta =2\alpha $
and thus, we find the general form for the Riemannian part $\mathbf{A}$ as:%
\begin{equation}
	\mathbf{A}=m\dot{t}^{2}-\dfrac{1}{1-\kappa_1 r^{2}}\dot{r}^{2}-r^{2}w^{2}.
\end{equation}%
Together with $\mathbf{B}=\dot{t},$ this leads to the general expressions
for the Finsler-Ricci scalar and for the scalar $a^{ab}R_{\cdot ab}$ as:  
\begin{equation}
	R=2\kappa_1 \left( \dfrac{\dot{r}^{2}}{\kappa_1 r^{2}-1}+r^{2}w^{2}\right)
	=2\kappa_1 \left( \mathbf{A}-\dfrac{\mathbf{B}^{2}}{\left\langle \mathbf{b,b}%
		\right\rangle }\right) ,~\ \ \ a^{ab}R_{\cdot ab}=12\kappa_1  \label{eq:R}
\end{equation}

Using the above expressions and $s = \frac{\mathbf{B}^2}{\mathbf{A}}$, we find that the right hand side of the field
equation reads, for Branch 1,%
\begin{equation}
	\dfrac{\mathbf{A}\left( a^{ab}R_{\cdot ab}\right) }{R}= \dfrac{2\kappa_1 }{1-\dfrac{s}{\left\langle \mathbf{b,b}%
			\right\rangle }}.  \label{eq:fieldeq_rhs_Branch1}
\end{equation}
\NV{We note that the denominators in both hand sides of the above are nonzero (apart from specific points that have to be excluded from the domain of definition of $L$). On the one hand, $R \neq 0$ was assumed from the very beginning and, on the other hand, $s =\left\langle \mathbf{b,b} \right\rangle$ cannot hold as this	would lead to $ \mathbf{A}=\tfrac{\mathbf{B}^{2}}{\left\langle \mathbf{b,b}\right\rangle }$, which is not possible, by virtue of \eqref{eq:Aexp}-\eqref{eq:Bexp}. }

Substitution of the latter expression into the field equation \eqref{Phi_ratio_R} then
transforms it into an ODE\ in the unknown function $\Psi =\Psi \left(
s\right) $.
	
\textbf{Remark:} Equation \eqref{eq:R} tells us, in particular, that
any vacuum solution belonging to this class has its underlying Riemannian
metric $\mathbf{A}$ with \emph{constant} \emph{nonzero} \emph{scalar curvature} - and \emph{absolutely parallel, non-null} 1-form $\mathbf{B}.$

\subsubsection{Integration of equation (\protect\ref{Phi_ratio_R})}
Using relations \eqref{eq:R}, the field equation \eqref{Phi_ratio_R}
can be rewritten as: $\Phi =\tfrac{3\mathbf{A}\left\langle \mathbf{b},%
	\mathbf{b}\right\rangle }{\mathbf{A}\left\langle \mathbf{b},\mathbf{b}%
	\right\rangle -\mathbf{B}^{2}},$ which is simply:
\begin{equation}
	\Phi =\dfrac{3\left\langle \mathbf{b},\mathbf{b}\right\rangle }{\left\langle 
		\mathbf{b},\mathbf{b}\right\rangle -s}\,. \label{eq:Phi_Case_1_simplified}
\end{equation}
Recall that 
\begin{equation}
	\Phi :=\left( 3\dfrac{\Psi -s\Psi ^{\prime }}{\Psi }+2\dfrac{s^{2}}{\mu }%
	\Psi ^{\prime \prime }\left( \Psi +\left\langle \mathbf{b,b}\right\rangle
	\Psi ^{\prime }-s\Psi ^{\prime }\right) \right)\,.
\end{equation}
\NV{We note that $\Phi$ is always well defined, as (apart from a conic subbundle of $TM$ which must be excluded from the domain of definition of $L$), $\mu$ cannot vanish (as discussed in Section \ref{sec:class3}) and $\Psi = 0 $ is not a valid choice -- as it would lead to an identically vanishing $L$. }
Thus, substituting $\Phi $ from the above, a direct computation shows
that this is equivalent (after discarding a common factor of $s$), to:
\begin{eqnarray}
	&&\left[ \Psi +\left( \left\langle \mathbf{b},\mathbf{b}\right\rangle
	-s\right) \Psi ^{\prime }\right] \times  \label{eq:Case_1_product} \\
	&&\left( 3s^{2}(\Psi ^{\prime })^{2}+3\Psi ^{2}+3\left\langle \mathbf{b},%
	\mathbf{b}\right\rangle \Psi \Psi ^{\prime }-6s\Psi \Psi ^{\prime
	}-3\left\langle \mathbf{b},\mathbf{b}\right\rangle s(\Psi ^{\prime
	})^{2}-4s^{2}\Psi \Psi ^{\prime \prime }+4\left\langle \mathbf{b},\mathbf{b}%
	\right\rangle s\Psi \Psi ^{\prime \prime }\right) =0.  \notag
\end{eqnarray}

Let us analyze each factor:
\begin{itemize}
	\item The situation $\Psi +\left( \left\langle \mathbf{b},\mathbf{b}%
	\right\rangle -s\right) \Psi ^{\prime }=0$ leads to $\Psi =c_{1}\left(
	\left\langle \mathbf{b},\mathbf{b}\right\rangle -s\right) ,$ $c_{1}=const.,$
	hence to a quadratic function:%
	\begin{equation}
		L=\mathbf{A}\Psi =c_{1}\left( \left\langle \mathbf{b},\mathbf{b}\right\rangle
		\mathbf{A}-\mathbf{B}^{2}\right) .
	\end{equation}%
	Yet, this is not a valid solution. A quick check shows that its associated
	Finslerian metric, with components:%
	\begin{equation}
		g_{ab}=c_{1}\left( \left\langle \mathbf{b},\mathbf{b}\right\rangle
		a_{ab}-b_{a}b_{b}\right) ,
	\end{equation}%
	has an identically vanishing determinant.
	
	\item It thus remains to integrate the following ODE:%
	\begin{equation}
		\left( 3s^{2}(\Psi ^{\prime })^{2}+3\Psi ^{2}+3\left\langle \mathbf{b},%
		\mathbf{b}\right\rangle \Psi \Psi ^{\prime }-6s\Psi \Psi ^{\prime
		}-3\left\langle \mathbf{b},\mathbf{b}\right\rangle s(\Psi ^{\prime
		})^{2}-4s^{2}\Psi \Psi ^{\prime \prime }+4\left\langle \mathbf{b},\mathbf{b}%
		\right\rangle s\Psi \Psi ^{\prime \prime }\right) =0.
	\end{equation}
\end{itemize}

This can be recast in the form: $\left( \Psi ^{-3}\left( \Psi -s\Psi
^{\prime }\right) ^{4}\left( \left\langle \mathbf{b},\mathbf{b}\right\rangle
-s\right) ^{3}\right) ^{\prime }=0,$ which leads to a first integral%
\begin{equation}
	\Psi ^{-3}\left( \Psi -s\Psi ^{\prime }\right) ^{4}\left( \left\langle 
	\mathbf{b},\mathbf{b}\right\rangle -s\right) ^{3}=-\tilde{c}_{1}
	\label{first_int_Case_1}
\end{equation}%
where $-\tilde{c}_{1}\in \mathbb{R}$ is a constant (and the minus sign is
chosen for convenience).

Further, the change of the dependent variable 
\begin{equation}
	f:=\dfrac{\Psi }{s},
\end{equation}%
casts equation (\ref{first_int_Case_1}) into the form%
\begin{equation}
	f^{\prime }=\tilde{c}_{1}f^{\tfrac{3}{4}}s^{-\tfrac{5}{4}}\left(
	s-\left\langle \mathbf{b},\mathbf{b}\right\rangle \right) ^{-\tfrac{3}{4}}.
\end{equation}%
This has the general solution%
\begin{equation}
	f=\dfrac{1}{s}\left[ c_{1}\left( s-\left\langle \mathbf{b},\mathbf{b}%
	\right\rangle \right) ^{\tfrac{1}{4}}+c_{2}s^{\tfrac{1}{4}}\right] ^{4},
\end{equation}%
where $c_{1}=\dfrac{\tilde{c}_{1}}{b}$ and $c_{2}\in \mathbb{R}$ are
arbitrary constants.

We note that, in order to have $s^{\tfrac{1}{4}}$ well defined, we will have
to work on regions where $A>0.$

\bigskip

Reverting to the original variable $\Psi =sf$ and introducing the factor $%
s>0 $ into the square bracket, we find:

\begin{proposition}
	The general solution for Branch 1 is:%
	\begin{equation}
		\Psi =\left[ c_{1}\left( s-\left\langle \mathbf{b},\mathbf{b}\right\rangle
		\right) ^{\tfrac{1}{4}}+c_{2}s^{\tfrac{1}{4}}\right] ^{4},
		\label{gen_sol_Branch_1}
	\end{equation} 
    where $c_{1}>0$ and $c_{2}\in \mathbb{R}$ are arbitrary constants.
	Equivalently:%
	\begin{equation}
		L=\left[ c_{1}\left( \mathbf{B}^{2}-\left\langle \mathbf{b},\mathbf{b}%
		\right\rangle \mathbf{A}\right) ^{\tfrac{1}{4}}+c_{2}\left\vert \mathbf{B}%
		\right\vert ^{\tfrac{1}{2}}\right] ^{4}.  \label{sol:L_Branch1}
	\end{equation}
\end{proposition}
\NV{In the conveniently chosen coordinates as in the previous section
\begin{align}
	L = \left[ c_2  |\dot t|^{\frac{1}{2}} + c_1 \left( \frac{1}{m(1-r^2 \kappa_1)}\dot r ^2 + \frac{r^2}{m} w^2\right)^{\frac{1}{4}}\right]^4\,,
\end{align}
where we recall that $m = \dfrac{1}{\langle \mathbf{b},\mathbf{b} \rangle}$ is a nonzero constant.}\\

\textbf{Remark: } \CP{A necessary condition for the existence of  light cones of $L$ is that the equation $L=0$ has nontrivial solutions, i.e., that $c_1$ and $c_2$ have opposite signs.}

\subsection{Branch 2:\ $a^{ab}R_{\cdot ab}=0$}\label{ssec:Branch2}
In this case, the field equation (\ref{Phi_ratio_R}) reads simply
\begin{equation}
	\Phi =0,  \label{eq_Case_2}
\end{equation}%
which allows us to directly integrate it.

Indeed, using (\ref{def_Phi}), this is:%
\begin{equation}
	3\dfrac{\Psi -s\Psi ^{\prime }}{\Psi }+2\dfrac{s^{2}}{\mu }\Psi ^{\prime
		\prime }\left( \Psi +\left\langle \mathbf{b,b}\right\rangle \Psi ^{\prime
	}-s\Psi ^{\prime }\right) =0,
\end{equation}%
where we recall that $\left\langle \mathbf{b,b}\right\rangle =const.$ and%
\begin{equation}
	\mu =\Psi \left( \Psi -s\Psi ^{\prime }\right) +\left( \left\langle \mathbf{%
		b,b}\right\rangle -s\right) \left( \Psi \Psi ^{\prime }+2s\Psi \Psi ^{\prime
		\prime }-s\Psi ^{\prime 2}\right) .
\end{equation}

Integration of the above via Maple leads to a general solution with two
sub-branches (denoted below by a $\pm $ sign):%
\begin{equation}\label{gen_sol_Branch_2}
	\Psi _{\pm }=K_{2}\left( s-\left\langle \mathbf{b},\mathbf{b}\right\rangle
	\right) \exp \left( -3\left\langle \mathbf{b},\mathbf{b}\right\rangle
	^{2}\int \dfrac{ds}{\left( s-\left\langle \mathbf{b},\mathbf{b}\right\rangle
		\right) \left( 2\left\langle \mathbf{b},\mathbf{b}\right\rangle s\pm \sqrt{%
			\left\langle \mathbf{b},\mathbf{b}\right\rangle ^{3}s+K_{1}s\left(
			s-\left\langle \mathbf{b},\mathbf{b}\right\rangle \right) }\right) }\right) ,
\end{equation}

where $K_{1},K_{2}\in \mathbb{R}$ are integration constants.

\bigskip

\textbf{Remark. }The choice $\left\langle \mathbf{b},\mathbf{b}\right\rangle
=0$ leads to $\Psi _{\pm }=K_{2}s$, i.e., to $L=K_{2}B^{2}$ which is
degenerate. Hence, in the above, we can assume with no loss of generality that%
\begin{equation}
	\left\langle \mathbf{b},\mathbf{b}\right\rangle \not=0.
\end{equation}

\NV{Also, we note that the values of $s$ for which $\Psi_{\pm}$ is ill-defined will correspond to points $(x,\dot{x}) \in TM$ which are be excluded from the domain of definition of $L$.}

It now remains to find $\mathbf{A}$ and $\mathbf{B}$.
This allows us to use, again Proposition \ref{prop:A_n} to set $\mathbf{B}=\dot{t}$ and to express $\mathbf{A}$. Thus,
the defining equality $a^{ab}R_{\cdot ab}=0$ for this Branch transforms, using (\ref{eq:scalar_curvature_n}), into 
\begin{equation}
	n_{,r}r=n^{2}+n;
\end{equation}%
this has the general solution%
\begin{equation}
	n=\dfrac{r}{\kappa _{2}-r},
\end{equation}%
where $\kappa _{2}\in \mathbb{R}$ is an arbitrary constant. Thus, in this
case,%
\begin{equation} \label{sol:A_Branch2}
	\mathbf{A}=m\dot{t}^{2}-\dfrac{r}{r-\kappa _{2}}\dot{r}^{2}-r^{2}w^{2},
\end{equation}%
where $m:=\dfrac{1}{\left\langle \mathbf{b,b}\right\rangle }$ is, again, a
constant.

\NV{This way, in the conveniently chosen coordinates defined in Section \ref{sec:class3}, the general solution for this Branch 2 is}

\begin{equation}
	L=\mathbf{A}\Psi _{\pm }\left(\tfrac{\mathbf{B}^2}{\mathbf{A}}\right),\label{sol:L_Branch2}
\end{equation}
where $\mathbf{B} = \dot{t}$ and $\mathbf{A}$ is given by \eqref{sol:A_Branch2}.

\section{Conclusion and outlook}\label{sec:CO}

We have identified \NV{the most interesting class (in view of a future study of asymptotic flatness) of spatially spherically symmetric ($SO\left(3\right) $-invariant) Berwald Finsler, as the only one that also contains flat metrics; for this class, we have found all exact, non-Ricci flat solutions of} the Finsler gravity equation \eqref{eq:field_eqn} in vacuum. There are three types of solutions, called $L_1$ and $L_\pm$, as displayed in \eqref{sol:L_Branch1} and \eqref{sol:L_Branch2}. Their derivation was discussed in detail in Section \ref{sec:FEQC3}.

A remarkable property is that all of these solutions belong to the class of $(\alpha,\beta)$-Finsler structures, that are built from a pseudo-Riemannian metric $\mathbf{a}$, \eqref{A}, and a $1$-form $\mathbf{b}$, \eqref{B}, which is absolutely parallel with respect to the Levi-Civita connection of $\mathbf{a}$. Moreover, for $L_1$, the scalar curvature of $\mathbf{a}$ is a non-zero constant, while for $L_\pm$, it vanishes.


With this, we have demonstrated for the first time that, in Finsler gravity,
even in the closest case to the Riemannian one -- represented by the
so-called Berwald-type metrics (whose defining property is that their geodesics are autoparallel curves of an affine connection on the spacetime
manifold) -- there exist vacuum spacetimes that are \emph{not} Ricci flat.
This is in sharp contrast with general relativity, where any vacuum solution
of the Einstein field equations must be Ricci flat.

So far, the collection of explicit solutions to the action-based Finsler gravity equation \eqref{eq:field_eqn} is not very large. With our findings here, we present new classes of explicit solutions, whose application to describe the gravitational field around physical compact objects will be the topic of a forthcoming paper.

\acknowledgments{
	\NV{The authors are deeply grateful to Dr. Fidel F. Villasenor, who has spotted two errors in a previous version of this paper, and to the anonymous JGP reviewer, who helped to substantially improve this manuscript.}
	
	The authors would like to acknowledge networking
	support by the COST Actions CA23130 \textquotedblleft Bridging high and low
	energies in search of quantum gravity (BridgeQG)\textquotedblright , CA21136
	\textquotedblleft Addressing observational tensions in cosmology with
	systematics and fundamental physics (CosmoVerse)\textquotedblright\ and
	CA24101 "Testing Fundamental Physics with Seismology" (FuSe). C.P. acknowledges support by the excellence cluster QuantumFrontiers of the German Research Foundation (Deutsche Forschungsgemeinschaft, DFG) under Germany's Excellence Strategy -- EXC-2123 QuantumFrontiers -- 390837967 and was funded by the Deutsche Forschungsgemeinschaft (DFG, German Research Foundation) - Project Number 420243324.
}

\appendix

\section{Connection coefficients and curvature components}\label{app:defs}
This section provides the definitions of the coefficients $k_i$ and $a_i$ used in Section \ref{ssec:classif_Berwald}.

The most general 4-dimensional, $SO(3)$-invariant, torsion-free affine connection $\Gamma$ on $M$ has, in spherical coordinates, the coefficients described by 12 smooth functions of $t$ and $r$, \cite{Cheraghchi:2022zgv}%
:%
\begin{align}
\Gamma _{tt}^{t}& =k_{1}(t,r), & \Gamma _{tr}^{t}& =k_{2}(t,r),  \notag \\
\Gamma _{rr}^{t}& =k_{3}(t,r), & \Gamma _{tt}^{r}& =k_{4}(t,r),  \notag \\
\Gamma _{rr}^{r}& =k_{5}(t,r), & \Gamma _{tr}^{r}& =k_{6}(t,r),  \notag \\
\Gamma _{\theta \theta }^{t}& =\tfrac{\Gamma _{\phi \phi }^{t}}{\sin
^{2}\theta }=k_{7}(t,r), & \Gamma _{\phi t}^{\phi }& =\Gamma _{\theta
t}^{\theta }=k_{8}(t,r),  \notag \\
\Gamma _{\phi r}^{\phi }& =\Gamma _{\theta r}^{\theta }=k_{9}(t,r), & \Gamma
_{\theta \theta }^{r}& =\tfrac{\Gamma _{\phi \phi }^{r}}{\sin ^{2}\theta }%
=k_{10}(t,r),  \notag \\
\sin \theta \Gamma _{t\theta }^{\phi }& =-\tfrac{\Gamma _{\phi t}^{\theta }}{%
\sin \theta }=k_{11}(t,r), & \Gamma _{\phi \phi }^{\theta }& =-\sin \theta
\cos \theta  \notag \\
\sin \theta \Gamma _{r\theta }^{\phi }& =-\tfrac{\Gamma _{r\phi }^{\theta }}{%
\sin \theta }=k_{12}(t,r), & \Gamma _{\theta \phi }^{\phi }& =\Gamma _{\phi
\theta }^{\phi }=\cot {\theta }\,.  \label{eq:appcon}
\end{align}
If the connection is Finsler metrizable, then there necessarily holds
\begin{equation}
	k_{11}=k_{12}=0,
\end{equation}
therefore, the most general Finsler metrizable, $SO(3)$-invariant connection is locally characterized, in spherical coordinates, by 10 smooth functions $k_i,\, i=1..10.$

The canonical nonlinear connection $N$ has the coefficients:%
\begin{equation}
N_{~b}^{a}=\Gamma _{~bc}^{a}x\,^{c}
\end{equation}%
and its curvature tensor $\mathcal{R}$ has the local components 
\begin{equation*}
R^{a}{}_{bc}=\delta _{c}N^{a}{}_{b}-\delta _{b}N^{a}{}_{c}=\partial
_{c}N^{a}{}_{b}-N^{d}{}_{c}\dot{\partial}_{d}N^{a}{}_{b}-\partial
_{b}N^{a}{}_{c}+N^{d}{}_{b}\dot{\partial}_{d}N^{a}{}_{c};
\end{equation*}%
in detail:%
\begin{equation}
\begin{array}{llll}
R_{~tr}^{t}=a_{1}\dot{t}+a_{2}\dot{r} & R_{~tr}^{r}=a_{3}\dot{t}+a_{4}\dot{r}
& R_{~tr}^{\theta }=a_{5}\dot{\theta} & R_{~tr}^{\phi }=a_{5}\dot{\phi} \\ 
R_{~t\theta }^{t}=a_{6}\dot{\theta} & R_{~t\theta }^{r}=a_{7}\dot{\theta} & 
R_{t\theta }^{\theta }=a_{8}\dot{t}+a_{9}\dot{r} & R_{t\theta }^{\phi }=0 \\ 
R_{~t\phi }^{t}=a_{6}\dot{\phi}\sin ^{2}\theta & R_{~t\phi }^{r}=a_{7}\dot{%
\phi}\sin ^{2}\theta & R_{~t\phi }^{\theta }=0 & R_{t\phi }^{\phi }=a_{8}%
\dot{t}+a_{9}\dot{r} \\ 
R_{~r\theta }^{t}=a_{10}\dot{\theta} & R_{~r\theta }^{r}=a_{11}\dot{\theta}
& R_{~r\theta }^{\theta }=a_{12}\dot{t}+a_{13}\dot{r} & R_{~r\theta }^{\phi
}=0 \\ 
R_{~r\phi }^{t}=a_{10}\dot{\phi}\sin ^{2}\theta & R_{~r\phi }^{r}=a_{11}\dot{%
\phi}\sin ^{2}\theta & R_{~r\phi }^{\theta }=0 & R_{~r\phi }^{\phi }=a_{12}%
\dot{t}+a_{13}\dot{r} \\ 
R_{~\theta \phi }^{t}=0 & R_{~\theta \phi }^{r}=0 & R_{~\theta \phi
}^{\theta }=-a_{14}\dot{\phi}\sin ^{2}\theta & R_{~\theta \phi }^{\phi
}=a_{14}\dot{\theta}\,.%
\end{array}%
,  \label{eq:appcurv}
\end{equation}%
where the quantities $a_{i}$ are functions of $t$ and $r$, expressed as: 
\begin{equation}
\begin{split}
a_{1}& =k_{1,r}-k_{2,t}+k_{3}k_{4}-k_{2}k_{6}\,, \\
a_{2}& =k_{2,r}-k_{3,t}+k_{2}^{2}+k_{3}k_{6}-k_{1}k_{3}-k_{2}k_{5}\,, \\
a_{3}& =k_{4,r}-k_{6,t}+k_{1}k_{6}+k_{4}k_{5}-k_{2}k_{4}-k_{6}^{2}\,, \\
a_{4}& =k_{6,r}-k_{5,t}+k_{2}k_{6}-k_{3}k_{4}\,, \\
a_{5}& =k_{8,r}-k_{9,t}\,, \\
a_{6}& =-k_{7,t}+k_{7}k_{8}-k_{1}k_{7}-k_{2}k_{10}\,, \\
a_{7}& =-k_{10,t}+k_{8}k_{10}-k_{4}k_{7}-k_{6}k_{10}\,, \\
a_{8}& =-k_{8,t}+k_{1}k_{8}+k_{4}k_{9}-k_{8}^{2}\,, \\
a_{9}& =-k_{9,t}+k_{2}k_{8}+k_{6}k_{9}-k_{8}k_{9}\,, \\
a_{10}& =-k_{7,r}+k_{7}k_{9}-k_{2}k_{7}-k_{3}k_{10}\,, \\
a_{11}& =-k_{10,r}+k_{9}k_{10}-k_{6}k_{7}-k_{5}k_{10}\,, \\
a_{12}& =-k_{8,r}+k_{2}k_{8}+k_{6}k_{9}-k_{8}k_{9}\,, \\
a_{13}& =-k_{9,r}+k_{3}k_{8}+k_{5}k_{9}-k_{9}^{2}\,, \\
a_{14}& =1+k_{7}k_{8}+k_{9}k_{10}\,.
\end{split}
\label{eq:a_i}
\end{equation}%
In the above, the subscripts $_{,t}$ and $_{,r}$ means partial
differentiation with respect to $t$ and $r$ respectively. The curvature
components $R_{b~cd}^{~a}$ of the affine connection are obtained from $%
R_{~cd}^{a}$ by $\dot{x}$-differentiation: 
\begin{equation}
R_{b~cd}^{~a}=\dot{\partial}_{b}R_{~cd}^{a};
\end{equation}%
e.g., $a_{1}=R_{t~tr}^{~t}$ etc.

\bigskip

Assume $k_{10}\not=0.$ Then, $a_{14}$ satisfies (\cite{Birkhoff-Th}, 
\cite{Voicu:2024mag} for more details): 
\begin{equation}
a_{14}=1+k_{10}^{2}\left( 2ab+c\right) ,~\partial
_{t}a_{14}=-2k_{10}a_{7}\left( c+2ab\right) ,\ \ \partial
_{r}a_{14}=-2k_{10}a_{11}\left( c+2ab\right) .  \label{eqs:a14}
\end{equation}%
where $a=\dfrac{k_{7}}{k_{10}},$ $b=\dfrac{k_{8}}{k_{10}},$ $c=\dfrac{%
k_{9}k_{10}-k_{7}k_{8}}{k_{10}^{2}}.$

Moreover, if the connection $\Gamma$ is metrizable by a non-Riemannian Finsler
function $L$, then it must obey:%
\begin{equation}
a_{6}=aa_{7},~~a_{8}=ba_{7},~\ a_{9}=\left( ab+c\right) a_{7},~\ \
a_{10}=aa_{11},~~a_{12}=ba_{11},~\ a_{13}=\left( ab+c\right) a_{11}.
\label{eq:metrizabillity_conds_a7}
\end{equation}

\section{Pseudo-Riemannian metric components\label{appx:A}}
In this Appendix we display the quantities derived from the pseudo-Riemannian part $\mathbf{a}$ (introduced in \eqref{A}) of a general Class 3 metric, that are needed for the calculations in Section \ref{sec:class3}.

Using $\mathbf{A}=ve^{2\mathcal{K}}+\left( e^{\mathcal{G}}\mathcal{M}\right)
u^{2}$ and $u=\dot{t}-a\dot{r},$ $v=c\dot{r}^{2}+2b\dot{t}\dot{r}-w^{2},$ we
find the explicit expression of $\mathbf{A}$ in terms of $\dot{t},\dot{r}$
and $w^{2}=\dot{\theta}^{2}+\dot{\varphi}^{2}\sin ^{2}\theta :$%
\begin{equation}
\mathbf{A}=\left( \mathcal{M}e^{\mathcal{G}}\right) \dot{t}^{2}+2\left( be^{2%
\mathcal{K}}-\mathcal{M}e^{\mathcal{G}}a\right) \dot{t}\dot{r}+\left( 
\mathcal{M}e^{\mathcal{G}}a^{2}+ce^{2\mathcal{K}}\right) \dot{r}^{2}-e^{2%
\mathcal{K}}w^{2}  \label{eq:A_trw}
\end{equation}%
Then, the components of the pseudo-Riemannian metric $a=\dfrac{1}{2}Hess_{%
\dot{x}}\mathbf{A}$ are%
\begin{equation}
\left( a_{ab}\right) =\left( 
\begin{array}{cccc}
\mathcal{M}e^{\mathcal{G}} & -\mathcal{M}e^{\mathcal{G}}a+e^{2\mathcal{K}}b
& 0 & 0 \\ 
-\mathcal{M}e^{\mathcal{G}}a+e^{2\mathcal{K}}b & \mathcal{M}e^{\mathcal{G}%
}a^{2}+e^{2\mathcal{K}}c & 0 & 0 \\ 
0 & 0 & -e^{2\mathcal{K}} & 0 \\ 
0 & 0 & 0 & -e^{2\mathcal{K}}\sin ^{2}\theta%
\end{array}%
\right) .
\end{equation}

\begin{itemize}
\item Determinant: 
\begin{equation}
\det \left( a_{ab}\right) =e^{6\mathcal{K}}\Delta \sin ^{2}\theta ,~\ \ \
\Delta :=\mathcal{M}e^{\mathcal{G}}\left( 2ab+c\right) -b^{2}e^{2\mathcal{K}%
}.  \label{def_Delta}
\end{equation}

\item Inverse metric components:%
\begin{equation}\label{eq:ainv}
\left( a^{ab}\right) =\left( 
\begin{array}{cccc}
\dfrac{1}{\Delta }\left( \mathcal{M}e^{\mathcal{G}-2\mathcal{K}%
}a^{2}+c\right) & \dfrac{1}{\Delta }\left( \mathcal{M}e^{\mathcal{G}-2%
\mathcal{K}}a-b\right) & 0 & 0 \\ 
\dfrac{1}{\Delta }\left( \mathcal{M}e^{\mathcal{G}-2\mathcal{K}}a-b\right) & 
\dfrac{1}{\Delta }\mathcal{M}e^{\mathcal{G}-2\mathcal{K}} & 0 & 0 \\ 
0 & 0 & -e^{-2\mathcal{K}} & 0 \\ 
0 & 0 & 0 & -\dfrac{e^{-2\mathcal{K}}}{\sin ^{2}\theta }%
\end{array}%
\right) .
\end{equation}
\end{itemize}
\bigskip

\bibliography{nonRicciSpherical}

\end{document}